%% file: paper.tex
\documentclass{fairmeta}
\usepackage{amsmath}
\usepackage{enumerate}
\usepackage{bm}
\usepackage{algorithm}
\usepackage{floatflt}
\usepackage{algpseudocode}
\usepackage{amsfonts}
\usepackage{amsthm}
\usepackage{newtxtt}
\usepackage{algorithm}
\usepackage{colortbl}
\usepackage{cleveref}
\usepackage{diagbox}
\usepackage[utf8]{inputenc}
\usepackage{textgreek}
\usepackage{colortbl}
\usepackage{nicematrix}
\usepackage{makecell}
\usepackage{float}
\usepackage{arydshln}
\usepackage[frozencache,cachedir=.]{minted}
\usepackage{caption}
\usepackage{subcaption}
\usepackage{tcolorbox}
\usepackage{amssymb}
\usepackage{xspace}
\usepackage{wrapfig}
\usepackage{adjustbox}
\usepackage{tabularx}
\usepackage{booktabs}
\usepackage{mathtools}
\usepackage{wrapfig}
\usepackage{amssymb}
\usepackage{graphicx}

\usepackage{amsmath,amssymb}
\usepackage[level=0]{wgroup_message}  

\usepackage{silence}
\makeatletter
\patchcmd{\wrong@fontshape}{\@gobbletwo}{}{}{}
\makeatother
\definecolor{upColor}{RGB}{17,138,21}
\definecolor{downColor}{RGB}{174,36,67}

\newtheorem{theorem}{Theorem}[]

\newtheorem{remark1}[theorem]{Remark}

\input{manuscript-commands.tex}

\title{Learning Graph Beliefs for Active Radio Mapping With Decision-Certified
Multi-Step Sensing}

\author[1]{Wangqian~Chen}
\author[2]{Hao~Wang}
\author[1]{Yumeng~Zhang}
\author[1]{Chenghong~Bian}
\author[2]{Junting~Chen}
\author[1]{Jun~Zhang}

\affiliation[1]{The Hong Kong University of Science and Technology}
\affiliation[2]{The Chinese University of Hong Kong, Shenzhen}

\metadata[Correspondence to]{Jun Zhang (\email{eejzhang@ust.hk})}

\begin{document}

\abstract{
Radio maps provide essential wireless propagation knowledge for network
optimization and autonomous platforms. Acquiring dense measurements,
however, is prohibitively costly, requiring active sampling to prioritize
informative measurements. Existing active sampling methods typically
rely on variances from prescribed spatial models or learned pointwise
uncertainty, which may not fully capture the spatial correlation of
different measurements. This paper proposes GraphBARC, a graph-structured
active radio map construction framework that integrates learned spatial
modeling with analytical Bayesian inference. A graph network learns
correlated node and edge factors using only acquired measurements,
inducing a joint full-map belief that supports analytical observation
updates without retraining. Rather than immediate one-step acquisition,
it adopts a nonmyopic sensing strategy that evaluates joint map-wide
Bayes risk reduction over feasible multi-step routes, accounting for
information redundancy among planned measurements. We reuse candidate
responses and derive residual-based bounds directly on route-value
differences, retaining shared response-error structure to avoid uniformly
high-accuracy evaluation. An error-aware continuation bound screens
unpromising prefixes, after which only responses needed for unresolved
comparisons are refined. We further establish a computable decision-gap
bound and show that certified exclusions remain valid during refinement.
Experiments show that GraphBARC reduces mean root-mean-square error
(RMSE) by over 16\% relative to the evaluated baselines and core planning
time by 56\% compared with the fixed-accuracy, unpruned implementation.
}

\maketitle

\section{Introduction}

Low-altitude wireless networks are expected to support
a wide range of aerial platforms, including \acpl{uav}, aerial robots,
and emerging air mobility terminals. These platforms require reliable
connectivity for navigation, trajectory planning, sensing, data collection,
and control \citep{WangSunZem:J25}. In urban environments, however,
wireless channels can vary significantly with location, building blockage,
\ac{los} availability, and surrounding scatterers. As a result, low-altitude
communication systems require accurate and timely awareness of the
radio environment. In this context, radio maps provide a natural representation
of wireless environmental knowledge, such as \ac{rss}, interference
level, beamforming gain, or low-dimensional \ac{csi} descriptors,
which can support connectivity-aware trajectory planning, coverage-hole
detection, beam management, and wireless digital twins \citep{ZenChenJ:24,LiBoChen:J24,WuZengJin:J24}.
Therefore, an accurate radio map can significantly improve the safety
and efficiency of autonomous aerial operations.

Constructing high-fidelity radio maps is challenging due to the prohibitive
costs of dense spatial sampling. Under a sensing budget, it is crucial
to prioritize highly informative measurements in complex environments.
While existing approaches, including interpolation techniques \citep{ShaoZhaLi:J26,SunChe:J24},
tensor completion \citep{TimiShre:J24,ZhaFu:J20}, and deep learning
models \citep{WangTaoNan:J25,ChenWChen:C25,TegaRomer:J22,JaenCair:J26,WangqianChen:J24},
can reconstruct radio maps from sparse observations, they typically
rely on predefined measurements that are sufficiently representative.
These \emph{passive} approaches do not actively decide where to measure
and thus fail to reduce measurement overhead. In online missions with
limited flight time, the lack of immediate feedback during passive
data collection can lead to redundant, low-quality, or uninformative
measurement efforts.

Active radio map construction integrates sensing and reconstruction
into a closed-loop process. As new measurements arrive, it updates
the current map estimate and uses current insights to guide subsequent
measurements. Compared with offline passive reconstruction, this active
setting is more challenging as it requires a principled mechanism
to quantify measurement informativeness to select high-utility samples.

Model-based Bayesian frameworks provide principled uncertainty estimates
for active radio map sensing \citep{ChenXiao:J2025,LinTao:J2026,ZhangZhouYou:C25,PolySadYe:J24,ShreRomeChe:J23}.
Gaussian process approaches \citep{ZhangZhouYou:C25,PolySadYe:J24}
have been widely adopted to facilitate sampling through analytical
posterior variances and specialized acquisition functions. The model-based
estimator in \citep{ShreRomeChe:J23} further enables online Bayesian
updates under the Gudmundson shadowing model. However, their fixed
or predefined kernels may inadequately capture heterogeneous and site-specific
spatial correlations in complex environments.

Deep learning approaches, such as UNet and Transformers \citep{ChenWangGuo:J25,ZheXiao:C25},
can learn spatial representations directly from data. While these
deterministic models can achieve high-resolution predictions, they
typically lack a principled uncertainty quantification for active
sensing. Deep reinforcement learning \citep{ZhaoHaoHua:J24} and trajectory
planning \citep{LiYaoAls:J26} have been developed to optimize UAV
measurement paths. However, these methods often demand extensive offline
pre-training or trial-and-error episodes, imposing heavy computational
or training burdens that hinder real-time deployment on resource-constrained
\ac{uav} platforms.

Generative models have enabled uncertainty-aware reconstruction through
ensembles of generated samples. Flow matching \citep{SunShicXiang:C26}
uses location-wise ensemble variance to guide candidate selection
and utility-aware path planning, while diffusion-based spectrum cartography
\citep{gu2025generat} combines posterior sampling with diversity-aware
sensing-location selection. Although these methods capture spatial
distributions, their iterative sampling procedures incur substantial
computational overhead for online deployment.

In summary, existing approaches face two key limitations. First, model-based
methods support analytical updates but rely on prescribed covariance
structures, whereas data-driven methods learn richer spatial structures,
yet often provide only deterministic predictions or marginal uncertainty.
Second, effective sensing requires evaluating the joint map-wide benefit
of multiple measurements under mobility constraints, rather than solely
one-step scores, as high local uncertainty does not necessarily imply
a large reduction in overall uncertainty. Some studies \citep{LuGao:C25,ChenZhuWang:J26}
have extended beyond one-step acquisition, but they require selecting
a destination before planning the route, which may overlook the joint
value of the endpoint and intermediate measurements. Joint route optimization
also entails combinatorial search and repeated posterior computations,
whose numerical errors may still affect route ranking.

To address these limitations, we develop GraphBARC, a graph-induced
Bayesian framework that learns a full-map belief with explicit cross-location
dependencies using nonmyopic active sensing. Newly acquired measurements
are assimilated analytically, with their information propagated across
correlated locations through the learned graph dependencies. Rather
than relying on one-step gains, GraphBARC evaluates joint map-wide
Bayes risk reduction over feasible multi-step routes. To reduce computational
cost, shared posterior responses are reused across routes, while residual-certified
bounds on route-value differences enable safe prefix pruning and selective
numerical refinement of unresolved comparisons. New observations immediately
update the belief, while periodic adaptation refines the graph factors
as more data accumulate.

The novelty and contributions are summarized as follows:
\begin{itemize}
\item We develop a graph-factor network that combines local message passing
and graph attention to learn map-wide propagation dependencies from
acquired measurements without external side information. Unlike deterministic
predictors, it induces a full-map belief that jointly characterizes
the map estimate and its uncertainty, while remaining compatible with
analytical Bayesian conditioning.
\item We derive bounds on the response inner-product and innovation-inverse
errors induced by approximate posterior solves, and use them to certify
route-value differences directly rather than separately certifying
each value. The resulting comparison intervals guide selective refinement
and can be tighter than those obtained by adding separate route-error
bounds under the same allowances.
\item We derive an error-aware continuation bound for safe prefix pruning
and establish a computable decision gap under a fixed belief and complete
route coverage. Certified prefix exclusions remain valid during subsequent
numerical refinement, allowing unresolved comparisons to be refined
without repeating the prefix search.
\item Experiments show that GraphBARC achieves over 16\% lower reconstruction
RMSE than the evaluated baselines while reducing core planning time
by 56\% relative to the fixed-accuracy, unpruned implementation.
\end{itemize}

The rest of the paper is organized as follows. The system model is
reviewed in Section II. Section III-V develop the graph-induced Bayesian
active learning framework and the corresponding implementation. Numerical
results are presented in Section VI, and Section VII concludes the
paper.

\begin{figure}[!t]
\centering\includegraphics[width=0.85\linewidth]{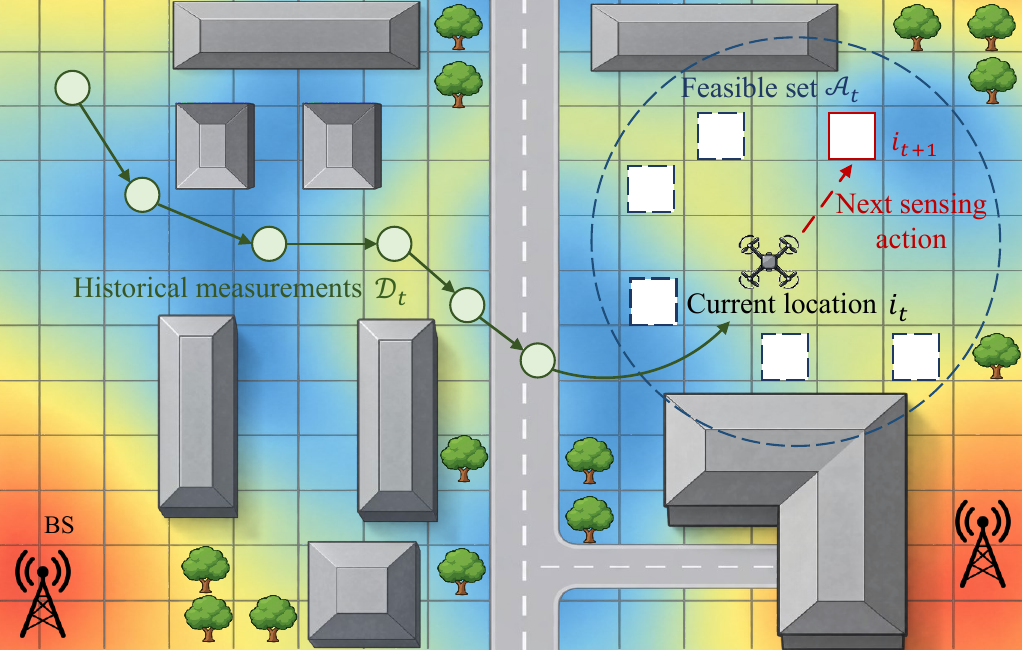}\caption{System scenario for active radio map construction, where a sensing
platform sequentially acquires measurements and selects the sensing
locations under sensing-budget constraints to improve overall reconstruction
accuracy.}
\label{fig:system_fig}
\end{figure}

\section{System Model and Problem Formulation}

We consider a low-altitude urban environment with multiple \acpl{bs}.
Let $\mathcal{X}\subset\mathbb{R}^{D}$ be the area of interest, where
$D$ is the spatial dimension. Discretizing $\mathcal{X}$ into $M$
grid cells gives the radio map $\mathbf{X}=[\mathbf{x}^{\mathrm{T}}_{1},...,\mathbf{x}^{\mathrm{T}}_{M}]\in\mathbb{R}^{M\times d}$,
where $\mathbf{x}_{i}\in\mathbb{R}^{d}$ represents the channel feature
associated with the $i$th cell, and $d$ is the feature dimension.
For a power map, $d=1$, and $\mathbf{x}_{i}$ reduces to the average
power value within the cell.

Let $\boldsymbol{\xi}=\mathrm{vec}(\mathbf{X}^{\mathrm{T}})\in\mathbb{R}^{Md}$,
where $\mathrm{vec}(\cdot)$ denotes column-wise vectorization. As
shown in Fig.~\ref{fig:system_fig}, a mobile platform visits selected
cells and acquires measurements to refine the current radio map. Let
$\mathbf{E}_{i}=\mathbf{e}^{\mathsf{T}}_{i}\otimes\mathbf{I}_{d}\in\mathbb{R}^{d\times Md}$,
where $\mathbf{e}_{i}\in\mathbb{R}^{M}$ is the $i$th canonical basis
vector, $\mathbf{I}_{d}$ is the $d\times d$ identity matrix, and
$\otimes$ denotes the Kronecker product. It follows that $\mathbf{E}_{i}\boldsymbol{\xi}=\mathbf{x}_{i}$.
The channel measurement collected at cell $i$ is thus given by
\begin{equation}
\mathbf{y}_{i}=\mathbf{E}_{i}\boldsymbol{\xi}+\mathbf{n}_{i},\hspace*{0.3cm}\mathbf{n}_{i}\sim\mathcal{N}(\mathbf{0},\mathbf{R}_{i}),
\end{equation}
where $\mathbf{R}_{i}\in\mathbb{R}^{d\times d}$ denotes the measurement
noise covariance.

Let $\mathcal{D}_{0}$ denote a small set of initialization measurements
acquired in the current scene. After $t$ active sensing actions,
the available measurement set is $\mathcal{D}_{t}=\mathcal{D}_{0}\cup\left\{ (i_{s},\mathbf{y}_{s}):1\leq s\leq t\right\} $.
Given the radio map estimator $f_{\boldsymbol{\theta}}(\cdot)$ with
parameters $\boldsymbol{\theta}$ and the sensing policy $\pi(\cdot)$,
the next location is selected by
\begin{equation}
i_{t+1}=\pi(\mathcal{D}_{t},f_{\boldsymbol{\theta}}(\mathcal{D}_{t}))\in\mathcal{A}_{t}\label{eq:channel_vector-2}
\end{equation}
where $\mathcal{A}_{t}$ denotes the feasible set at step $t$. Consider
a sensing budget $\mathit{\Gamma}$ and an available travel distance
$\ell_{\max}$, the active radio map construction problem is formulated
as
\begin{equation}
\begin{array}{c}
(\pi^{*},\boldsymbol{\theta}^{*})\in\mathrm{\underset{\pi,\boldsymbol{\theta}}{\textrm{argmin}}}\hspace{0.3cm}\mathbb{E}[\mathcal{L}(f_{\boldsymbol{\theta}}(\mathcal{D}^{\pi}_{\varGamma}),\boldsymbol{\xi})]\\
\hspace*{2cm}\mathrm{s.t.}\quad\sum^{\varGamma}_{t=1}c_{t}\leq\ell_{\max},
\end{array}
\end{equation}
where $\mathcal{D}^{\pi}_{\varGamma}$ denotes the measurement history
collected after $\varGamma$ sensing steps under policy $\pi(\cdot)$,
$\mathcal{L}(\cdot)$ is the reconstruction loss, and $c_{t}$ denote
the corresponding travel distance at step $t$.

Active radio map construction aims at identifying informative measurements
to minimize the final reconstruction error. This task thus involves
the joint optimization of radio map reconstruction and measurement-location
selection.

\section{Graph-Structured Belief and Inference\label{sec:Graph-Structured-Belief}}

Active radio mapping requires a spatial representation that supports
both reconstruction and sensing decisions. The challenge is to adapt
spatial dependencies to the acquired measurements while retaining
an explicit mechanism for assimilating new observations and evaluating
future measurements. We construct a graph-structured belief that captures
spatial dependencies, whose induced covariance governs both analytical
measurement updates and map-wide Bayes-risk reduction, providing a
unified basis for reconstruction and active sensing.

\subsection{Learned Graph-Induced Full-Map Belief}

A local observation informs unmeasured locations only if the model
captures their spatial dependence. Classical interpolation kernels
or predefined Gaussian models may be restrictive in complex environment
that affected by blockages and scatterers. While purely data-driven
neural estimators learn complex spatial mappings, they lack a reliable
uncertainty estimates or an analytic full-map belief for active sampling.

Our key idea is to use the network to learn a graph, and model the
radio map as a graph-structured random field. Let $\mathcal{G}=(\mathcal{V},\mathcal{E})$
be a graph over $M$ cells, where each node $i\in\mathcal{V}$ corresponds
to one cell, and edge $(i,j)\in\mathcal{E}$ represents a structural
relation from cell $i$ and $j$. The graph captures spatial dependencies
learned from the acquired measurements.

For an acquired context $\mathcal{D}_{\tau}$, it learns a node prior
$\mathbf{m}_{i}\in\mathbb{R}^{d}$ and a precision matrix $\boldsymbol{\Omega}_{i}\in\mathbb{R}^{d\times d}$.
The node factor is
\begin{equation}
\psi_{i}(\mathbf{x}_{i};\mathcal{D}_{\tau})=\exp[-\frac{1}{2}(\mathbf{x}_{i}-\mathbf{m}_{i})^{\mathsf{T}}\mathbf{\Omega}_{i}(\mathbf{x}_{i}-\mathbf{m}_{i})].\label{eq:graph_node}
\end{equation}
Here, $\mathbf{m}_{i}$ specifies the local estimation tendency at
cell $i$, and $\mathbf{\Omega}_{i}\succ0$ measures the confidence
of this prior.

For an edge $(i,j)\in\mathcal{E}$, we define an edge relation as
\begin{equation}
\adjustbox{max width=.91\linewidth}{\ensuremath{\varphi_{ij}(\mathbf{x}_{i},\mathbf{x}_{j};\mathcal{D}_{\tau})=\exp[-\frac{1}{2}(\mathbf{x}_{i}-\mathbf{M}_{ij}\mathbf{x}_{j})^{\mathsf{T}}\mathbf{W}_{ij}(\mathbf{x}_{i}-\mathbf{M}_{ij}\mathbf{x}_{j})],}}\label{eq:graph_edge}
\end{equation}
where $\mathbf{M}_{ij}\in\mathbb{R}^{d\times d}$ describes the correlation
from cell $i$ to cell $j$, and $\mathbf{W}_{ij}\in\mathbb{R}^{d\times d}$
is a positive edge precision. Because these factors are context dependent,
the learned dependence can vary across location, direction, and feature
dimensions.

Combining these factors gives the full-map belief
\begin{equation}
q_{\tau}(\boldsymbol{\xi};\mathcal{D}_{\tau})\propto\prod^{M}_{i=1}\psi_{i}(\mathbf{E}_{i}\boldsymbol{\xi};\mathcal{D}_{\tau})\prod_{(i,j)\in\mathcal{E}}\varphi_{ij}(\mathbf{E}_{i}\boldsymbol{\xi},\mathbf{E}_{j}\boldsymbol{\xi};\mathcal{D}_{\tau}).\label{eq:graph_prior}
\end{equation}
This factorization combines local information and cross-cell structural
dependence. The node factors describe what each cell is expected to
be, while the edge factors describe how information should propagate
between related cells. Thus, the graph belief provides a probabilistic
mechanism for propagating information beyond the set of directly measured
locations.
\begin{prop}
\label{thm:Proposed1} (Graph-induced full-map belief). Suppose that
$\mathbf{\Omega}_{i}\succ\mathbf{0}$, $\forall i\in\mathcal{V}$
and $\mathbf{W}_{ij}\succeq\mathbf{0}$, $\forall(i,j)\in\mathcal{E}$.
Define
\begin{equation}
\boldsymbol{\Lambda}_{\tau}=\sum^{M}_{i=1}\mathbf{E}^{\mathsf{T}}_{i}\mathbf{\Omega}_{i}\mathbf{E}_{i}+\sum_{(i,j)\in\mathcal{E}}\mathbf{A}^{\mathsf{T}}_{ij}\mathbf{W}_{ij}\mathbf{A}_{ij}\label{eq:Lambda_g}
\end{equation}
and
\begin{equation}
\boldsymbol{\eta}_{\tau}=\sum^{M}_{i=1}\mathbf{E}^{\mathsf{T}}_{i}\mathbf{\Omega}_{i}\mathbf{m}_{i},\hspace*{0.2cm}\mathbf{A}_{ij}\triangleq\mathbf{E}_{i}-\mathbf{M}_{ij}\mathbf{E}_{j}.\label{eq:Lambda_g-1}
\end{equation}
Then, $\boldsymbol{\Lambda}_{\tau}\succ\mathbf{0}$, and the graph-induced
belief in (\ref{eq:graph_prior}) defines the proper Gaussian distribution
\begin{equation}
q_{\tau}(\boldsymbol{\xi};\mathcal{D}_{\tau})=\mathcal{N}(\boldsymbol{\xi};\boldsymbol{\mu}_{\tau},\boldsymbol{\Sigma}_{\tau}),\label{eq:canonical_prior}
\end{equation}
where
\begin{equation}
\boldsymbol{\mu}_{\tau}=\boldsymbol{\Lambda}^{-1}_{\tau}\boldsymbol{\eta}_{\tau},\hspace*{0.2cm}\boldsymbol{\Sigma}_{\tau}=\boldsymbol{\Lambda}^{-1}_{\tau}.\label{eq:AOD_model-1-1}
\end{equation}
Here, $\boldsymbol{\mu}_{\tau}\in\mathbb{R}^{Md}$ and $\boldsymbol{\Sigma}_{\tau}\in\mathbb{R}^{Md\times Md}$
are the mean and covariance induced by the fixed graph factors.
\end{prop}
\begin{proof}
See Appendix~\ref{app:Proposition 1 proof}.
\end{proof}
Proposition~\ref{thm:Proposed1} connects factor learning to a valid
full-map belief. The first term in $\boldsymbol{\Lambda}_{\tau}$
encodes node-wise confidence, while the second introduces cross-cell
coupling. Although $\boldsymbol{\Lambda}_{\tau}$ is sparse, its inverse
is generally dense within connected components. The diagonal block
$\boldsymbol{\Sigma}_{\tau,ii}$ quantify cell-wise uncertainty, whereas
the off-diagonal blocks $\boldsymbol{\Sigma}_{\tau,ij}$ encode cross-location
dependencies, which are the key quantities governing both measurement
propagation and sensing value.

\subsection{Measurement Responses and Multi-step Sensing Value\label{subsec:Bayesian-Posterior-Update}}

Active radio map construction requires each new measurements to be
immediately incorporated into both the map estimate and its uncertainty.
Conventional neural estimators, however, typically lack an explicit
probabilistic mechanism for incremental conditioning and map-wide
uncertainty propagation. We condition the learned graph-induced belief
on the acquired measurements via analytical Bayesian updates and use
the same posterior response to evaluate candidate value.

Let $\boldsymbol{\mu}_{t}$ and $\boldsymbol{\Sigma}_{t}$ be the
mean and covariance at sensing step $t$. For a new observation $\mathbf{y}_{t+1}$
at cell $j$, define $\mathbf{C}_{t,j}=\boldsymbol{\Sigma}_{t}\mathbf{E}^{\mathsf{T}}_{j}$
and $\mathbf{S}_{t,j}=\mathbf{E}_{j}\mathbf{C}_{t,j}+\mathbf{R}_{j}$.
Gaussian conditioning gives
\begin{equation}
\begin{aligned}\boldsymbol{\mu}^{+}_{t} & =\boldsymbol{\mu}_{t}+\mathbf{C}_{t,j}\mathbf{S}_{t,j}{}^{-1}\bigl(\mathbf{y}_{t+1}-\mathbf{E}_{j}\boldsymbol{\mu}_{t}\bigr),\\
\boldsymbol{\Sigma}^{+}_{t} & =\boldsymbol{\Sigma}_{t}-\mathbf{C}_{t,j}\mathbf{S}_{t,j}{}^{-1}\mathbf{C}^{\mathsf{T}}_{t,j}.
\end{aligned}
\label{eq:update}
\end{equation}
The superscript $+$ denotes conditioning on the new observation while
retaining the current graph factors. Each action adds a localized
positive information block without recomputing the graph factors at
every sensing step.

To quantify the map-wide sensing value, let $\mathbf{W}\in\mathbb{R}^{Md\times Md}$
be a positive-semidefinite task-weight matrix satisfying trace $\mathrm{tr}(\mathbf{W})=1$.
The map-wise posterior Bayes risk is
\begin{equation}
\mathcal{R}_{t}=\mathbb{E}_{q_{t}}\left[\|\boldsymbol{\xi}-\boldsymbol{\mu}_{t}\|^{2}_{\mathbf{W}}\right]=\mathrm{tr}(\mathbf{W}\boldsymbol{\Sigma}_{t}),
\end{equation}
where $\|\mathbf{z}\|^{2}_{\mathbf{W}}=\mathbf{z}^{\mathsf{T}}\mathbf{W}\mathbf{z}$.
A common choice $\mathbf{W}=\mathbf{I}_{Md}/(Md)$ gives average squared
error per cell and feature component.

The one-step risk reduction obtained by measuring cell $j$ is
\begin{equation}
\begin{aligned}\Delta_{t}(j) & =\mathcal{R}_{t}-\mathcal{R}^{+}_{t}(j)=\mathrm{tr}(\mathbf{S}_{t,j}{}^{-1}\mathbf{C}^{\mathsf{T}}_{t,j}\mathbf{W}\mathbf{C}_{t,j})\geq0\end{aligned}
.\label{eq:one-step Bayes risk}
\end{equation}
It follows the same posterior response used for belief refinement.
However, selecting sensing locations solely by these immediate gains
may be myopic. We therefore jointly evaluate prospective measurements
along feasible multi-step routes.

At planning step $t$, consider a hypothetical feasible route $P_{t}=(j_{t,1},\ldots,j_{t,h_{t}})$
with $h_{t}$ prospective locations. Let $P_{t,1:k}$ denotes its
first $k$ locations and $\boldsymbol{\Sigma}^{P_{t,1:k}}_{t}$ the
corresponding counterfactual posterior covariance. For an unvisited
cell $j$, its branch-conditioned one-step gain is
\begin{equation}
\Delta_{t}(j\mid P_{t,1:k})=\mathrm{tr}((\mathbf{S}^{P_{t,1:k}}_{t,j}){}^{-1}(\mathbf{C}^{P_{t,1:k}}_{t,j})^{\mathsf{T}}\mathbf{W}\mathbf{C}^{P_{t,1:k}}_{t,j}).
\end{equation}
The prospective terminal route value consequently satisfies
\begin{equation}
J_{t}(P_{t})=\mathcal{R}_{t}-\mathrm{tr}(\mathbf{W}\boldsymbol{\Sigma}^{P_{t}}_{t})=\sum^{h_{t}}_{k=1}\Delta_{t}(j_{k}\mid P_{t,1:k-1}).\label{eq:value}
\end{equation}
The equality follows by telescoping successive risk differences. Each
measurement is valued after accounting for the information already
planned on that branch.

Let $\mathbf{p}^{\mathrm{u}}_{t}$ be the current location and $B_{t}=(\mathit{\Gamma}_{t},\ell_{t})$
the remaining budgets. Let $\mathcal{P}_{t}=\mathcal{P}(\mathcal{A}_{t},\mathbf{p}^{\mathrm{u}}_{t},B_{t})$
be the feasible routes set under the constraints. The multi-step sensing
objective is
\begin{equation}
P^{\star}_{t}\in\arg\max_{P\in\mathcal{P}_{t}}J_{t}(P_{t}).\label{eq:long-horizon}
\end{equation}
This objective selects measurements jointly according to their terminal
map-wide improvement, rather than independently maximizing immediate
gains. A smaller immediate gain can be preferred when it enables a
more informative feasible batch.

Let $\widehat{P}_{t}=(\widehat{j}_{t,1},\ldots,\widehat{j}_{t,h_{t}})$
be the final route returned by the planner. All selected locations
are visited in order:
\[
i_{t}=\widehat{j}_{t,h_{t}},\hspace*{0.1cm}\mathcal{D}_{t+h_{t}}=\mathcal{D}_{t}\cup\left\{ (i_{t+k},\mathbf{y}_{t+k}):k=1,\ldots,h_{t}\right\} .
\]
The current location is updated to the route endpoint. This planning
thus enables nonmyopic sensing by considering the cumulative information
gain among multi-step future choices.

\section{Decision-Certified Multi-Step Sensing}

Section~\ref{sec:Graph-Structured-Belief} defines the route value
over a multi-step lookahead horizon. Finding a high-value route remains
computationally demanding due to the combinatorial growth of candidate
routes and the costly posterior-response solves. Uniformly accurate
response solves may waste computation on comparisons that are already
decisive, whereas insufficient accuracy can reverse route rankings
and invalidate pruning.

We address these through shared posterior responses and decision-level
error control. Candidate responses are reused, and their residuals
yield bounds directly on route-value differences. A continuation bound
then excludes uncompetitive prefixes, after which numerical refinement
focuses on unresolved comparisons. Throughout a planning call, the
graph belief, task weights, noise model, and candidate domain remain
fixed.

\subsection{Shared Conditional-Response Representation}

Two routes visiting the same locations with different orders have
the same terminal value. Even routes with different measurement sets
can reuse their common candidate responses.

For a feasible route $P$, stack its block-selection matrices into
$\mathbf{E}_{P}\in\mathbb{R}^{|P|d\times Md}$, and define
\begin{equation}
\mathbf{C}^{P}_{t}=\boldsymbol{\Sigma}_{t}\mathbf{E}^{\mathsf{T}}_{P},\hspace*{0.2cm}\mathbf{S}^{P}_{t}=\mathbf{E}_{P}\mathbf{C}^{P}_{t}+\mathbf{R}_{P},
\end{equation}
where $\mathbf{R}_{P}=\mathrm{blkdiag}(\mathbf{R}_{j}:j\in P)$ is
the block-diagonal noise matrix. Gaussian conditioning gives
\begin{equation}
J_{t}(P_{t})=\mathrm{tr}((\mathbf{S}^{P_{t}}_{t})^{-1}\mathbf{G}^{P_{t}}_{t}),\hspace*{0.2cm}\mathbf{G}^{P_{t}}_{t}=(\mathbf{C}^{P_{t}}_{t})^{\mathsf{T}}\mathbf{W}\mathbf{C}^{P_{t}}_{t}.\label{eq:full-value}
\end{equation}
The candidate responses $\mathbf{C}_{t,j}$ and Gram blocks $\mathbf{G}_{t,ij}=\mathbf{C}^{\mathsf{T}}_{t,i}\mathbf{W}\mathbf{C}_{t,j}$
are computed once and reused across route evaluations. Routes visiting
the same candidate set have identical terminal values because reordering
does not change the trace.

The same representation supports conditional queries. For a prefix
$P$ and a candidate $j\notin P$,
\begin{equation}
\mathbf{C}^{P}_{t,j}=\mathbf{C}_{t,j}-\mathbf{C}^{P}_{t}\left(\mathbf{S}^{P}_{t}\right)^{-1}\mathbf{E}_{P}\mathbf{C}_{t,j}.
\end{equation}
The second term removes the response component already explained by
the planned prefix. Different prefixes therefore require only small
conditional updates rather than independent full-map conditioning.
Together, the order-invariance property and shared response representation
avoid repeated full-map response solves and duplicate value evaluations.

\subsection{Residual-Certified Decision Bounds\label{subsec:Residual-Based-Decision-Bounds}}

The shared representation reduces route evaluation to small matrix
operations. However, approximate response solves perturb both the
Gram blocks and the inverse innovation matrices. We first bound these
effects, and then combine shared Gram errors before applying norm
bounds to competing routes.

Let $\widehat{\mathbf{C}}_{t,j}$ approximate the response satisfying
$\boldsymbol{\Lambda}_{t}\mathbf{C}_{t,j}=\mathbf{E}^{\mathsf{T}}_{j}$.
For a verified bound $\boldsymbol{\Lambda}_{t}\succeq\alpha_{t}\mathbf{I}_{Md}$,
with $\alpha_{t}>0$, define
\begin{equation}
\epsilon_{t,j}=\frac{||\mathbf{E}^{\mathsf{T}}_{j}-\boldsymbol{\Lambda}_{t}\widehat{\mathbf{C}}_{t,j}||_{F}}{\alpha_{t}},\hspace*{0.2cm}\epsilon_{t}(P)=(\sum_{j\in P}\epsilon^{2}_{t,j})^{1/2},
\end{equation}
where $||\cdot||_{F}$ is the Frobenius norm. These residual-based
quantities bound the errors of individual and stacked response blocks.
Because approximate responses need not induce a symmetric positive-semidefinite
signal covariance, define
\[
\widehat{\mathbf{S}}^{P}_{t}=\mathrm{sym}(\mathbf{E}_{P}\widehat{\mathbf{C}}^{P}_{t})+\mathbf{R}_{P}+\gamma_{t}\mathbf{I},
\]
where $\mathrm{sym}(\mathbf{A})=(\mathbf{A}+\mathbf{A}^{\mathsf{T}})/2$,
and $\gamma_{t}\ge0$ is a correction to make the numerical covariance
positive semidefinite, yielding $\widehat{\mathbf{S}}^{P}_{t}\succeq\mathbf{R}_{P}\succ\mathbf{0}$
for every considered route $P$.

Let $\widehat{J}_{t}(P)$ be the approximate route value obtained
by numerical solvers. Its error admits the decomposition
\begin{equation}
\begin{aligned}J_{t}(P)-\widehat{J}_{t}(P)= & \underbrace{\mathrm{tr}((\widehat{\mathbf{S}}^{P}_{t})^{-1}(\mathbf{G}^{P}_{t}-\widehat{\mathbf{G}}^{P}_{t}))}_{\text{Gram error}}\\
 & +\underbrace{\mathrm{tr}(((\mathbf{S}^{P}_{t})^{-1}-(\widehat{\mathbf{S}}^{P}_{t})^{-1})\mathbf{G}^{P}_{t}).}_{\text{innovation-inverse error}}
\end{aligned}
\label{eq:decomposition}
\end{equation}
It separates the two error sources and allows shared Gram errors to
be combined across competing routes before bounding, yielding tighter
route-value comparison intervals.
\begin{prop}
\label{prop:Residual-based-conditioning}(Residual-Based Route Error
Bounds). Suppose that $\boldsymbol{\Lambda}_{t}\succeq\alpha_{t}\mathbf{I}_{Md}$,
\textup{$\widehat{\mathbf{S}}^{P}_{t}\succeq\mathbf{R}_{P}\succ\mathbf{0}$.}
\textup{Let $s_{P}$, $\widehat{s}_{P}$} be valid positive lower
bounds on the smallest eigenvalues of \textup{$\mathbf{S}^{P}_{t}$,
$\widehat{\mathbf{S}}^{P}_{t}$}, respectively, and $\omega=\sqrt{||\mathbf{W}||{}_{2}}$.
Consequently, \textup{
\begin{equation}
\|(\mathbf{S}^{P}_{t})^{-1}-(\widehat{\mathbf{S}}^{P}_{t})^{-1}\|_{2}\leq\frac{\epsilon_{t}(P)+\gamma_{t}}{s_{P}\widehat{s}_{P}},\label{eq:conditioning stability}
\end{equation}
}where $\|\cdot\|_{2}$ is the spectral norm. Furthermore, the Gram-block
errors and the innovation-inverse error satisfy
\begin{equation}
\begin{aligned}\beta_{t,ij} & \geq\|\mathbf{G}_{t,ij}-\widehat{\mathbf{G}}_{t,ij}\|_{F},\\
\kappa_{t}(P) & \geq\left|\mathrm{tr}(((\mathbf{S}^{P}_{t})^{-1}-(\widehat{\mathbf{S}}^{P}_{t})^{-1})\mathbf{G}^{P}_{t})\right|,
\end{aligned}
\label{eq:Two_bounded}
\end{equation}
\textup{where }
\begin{equation}
\begin{aligned}\beta_{t,ij} & =\omega(\epsilon_{t,i}\|\mathbf{W}^{1/2}\widehat{\mathbf{C}}_{t,j}\|_{F}+\epsilon_{t,j}\|\mathbf{W}^{1/2}\widehat{\mathbf{C}}_{t,i}\|_{F})\\
 & \hspace*{0.4cm}+\omega^{2}\epsilon_{t,i}\epsilon_{t,j},\\
\kappa_{t}(P) & =\frac{\epsilon_{t}(P)+\gamma_{t}}{s_{P}\widehat{s}_{P}}(\|\mathbf{W}^{1/2}\widehat{\mathbf{C}}^{P}_{t}\|_{F}+\omega\epsilon_{t}(P)){}^{2}.
\end{aligned}
\label{eq:Two_bounded-1}
\end{equation}
\end{prop}
\begin{proof}
See Appendix~\ref{app:Proposition 2 proof}.
\end{proof}
Proposition~\ref{prop:Residual-based-conditioning} provides residual-certified
bounds for the two error sources in numerical route valuation. These
bounds are combined next to directly certify route-value differences.
\begin{cor}
\label{lem:(Shared-error-route-comparisons)}(Shared-error route comparisons).
Under Proposition~\ref{prop:Residual-based-conditioning}, suppose
that both routes $P$ and $Q$ use the same shared approximate response
and Gram blocks. Define
\begin{equation}
\begin{aligned}e_{t}(P,Q)= & \sum_{i,j\in P\cup Q}\left\Vert [(\widehat{\mathbf{S}}^{P}_{t})^{-1}]_{ij}-[(\widehat{\mathbf{S}}^{Q}_{t})^{-1}]_{ij}\right\Vert _{F}\beta_{t,ij}\\
 & +\kappa_{t}(P)+\kappa_{t}(Q),
\end{aligned}
\end{equation}
where $P\cup Q$ denotes the union of two routes, $[(\widehat{\mathbf{S}}^{P}_{t})^{-1}]{}_{ij}$
is the $d\times d$ block associated with locations $i$ and $j$,
and is interpreted as zero if either location is absent from $P$.
Then
\begin{equation}
\left|J_{t}(P)-J_{t}(Q)-\left(\widehat{J}_{t}(P)-\widehat{J}_{t}(Q)\right)\right|\le e_{t}(P,Q).
\end{equation}
\end{cor}
\begin{proof}
See Appendix~\ref{Proof of Corollary 1}.
\end{proof}
Corollary~\ref{lem:(Shared-error-route-comparisons)} shows that
route preferences can be certified by directly bounding their value
differences, without requiring equally accurate estimates of individual
route values. Similar inverse-innovation coefficients can yield a
reliable comparison even when the individual route values remain less
accurately resolved. Define the decision interval by
\begin{equation}
\begin{aligned}[][\underline{\Delta}_{t}(P,Q),\,\overline{\Delta}_{t}(P,Q)] & =[\widehat{J}_{t}(P)-\widehat{J}_{t}(Q)\pm e_{t}(P,Q)].\end{aligned}
\label{eq:interval}
\end{equation}
Thus, a positive lower bound certifies that $P$ improves upon $Q$;
a nonpositive upper bound certifies that it cannot improve. Otherwise,
the comparison remains unresolved and requires further numerical response
refinement. Refinement updates until the decision is resolved or the
stopping criterion is met.

\subsection{Prefix Screening and Decision Refinement}

The preceding intervals compare two routes under specified numerical
accuracy. Prefix screening requires another easily acquired bound
covering every feasible continuation, including routes not yet generated.
We obtain this bound from conditional response errors and then separate
geometric exploration from subsequent numerical refinement.

A feasible greedy route $Q^{0}_{t}$ is constructed as a fixed reference
for screening. For a nonempty prefix $P$ and $j\notin P$, let $\delta_{t,j}(P)$
be a computable bound (see Appendix~\ref{Conditional-Response Error})
satisfying $||\mathbf{W}^{1/2}(\mathbf{C}^{P}_{t,j}-\widehat{\mathbf{C}}^{P}_{t,j})||_{F}\le\delta_{t,j}(P)$.
Define
\begin{equation}
\overline{b}_{t,j}(P)=\frac{(\|\mathbf{W}^{1/2}\widehat{\mathbf{C}}^{P}_{t,j}\|_{F}+\delta_{t,j}(P))^{2}}{\lambda_{\min}(\mathbf{R}_{j})}
\end{equation}
where $\lambda_{\min}(\cdot)$ denotes the smallest eigenvalue.
\begin{lem}
\label{lem:(Error-aware-continuation-bound)}(Error-aware continuation
bound). Let $\mathcal{A}_{t}(P)$ contain every location that may
appear in a feasible continuation of $P$, and let $r_{P}=h_{t}-|P|$.
For any feasible continuation $T$,
\begin{equation}
\adjustbox{max width=.88\linewidth}{\ensuremath{J_{t}(P\oplus T)-J_{t}(P)\leq\operatorname{TopSum}_{r_{P}}\{\overline{b}_{t,j}(P):j\in\mathcal{A}_{t}(P)\},}}\label{eq:graph_edge-1}
\end{equation}
where $\oplus$ denotes concatenation, and $\operatorname{TopSum}_{r}\{\cdot\}$
sums the largest at most $r$ nonnegative entries.
\end{lem}
\begin{proof}
See Appendix~\ref{Proof of Lemma 1}.
\end{proof}
Lemma~\ref{lem:(Error-aware-continuation-bound)} provides an optimistic
upper bound on the maximum additional value that any feasible continuation
of a prefix can achieve. It relaxes joint innovation effects and motion
interactions, yielding a cheap bound that is sufficient for branch
pruning without assuming diminishing returns. Define
\begin{equation}
\adjustbox{max width=.89\linewidth}{\ensuremath{U_{t}(P;Q^{0}_{t})=\overline{\Delta}_{t}(P,Q^{0}_{t})+\operatorname{TopSum}_{r_{P}}\{\overline{b}_{t,j}(P):j\in\mathcal{A}_{t}(P)\}.}}\label{eq:bound}
\end{equation}
The first term bounds the prefix value relative to the reference;
the second bounds its possible additional value. Consequently, $J_{t}(P\oplus T)-J_{t}(Q^{0}_{t})\leq U_{t}(P;Q^{0}_{t})$
for every feasible continuation. A prefix with $U_{t}(P;Q^{0}_{t})\leq0$
is discarded. Otherwise, its successors are generated from the current
endpoint.

\begin{figure}[!t]
\centering\includegraphics[width=1.0\linewidth]{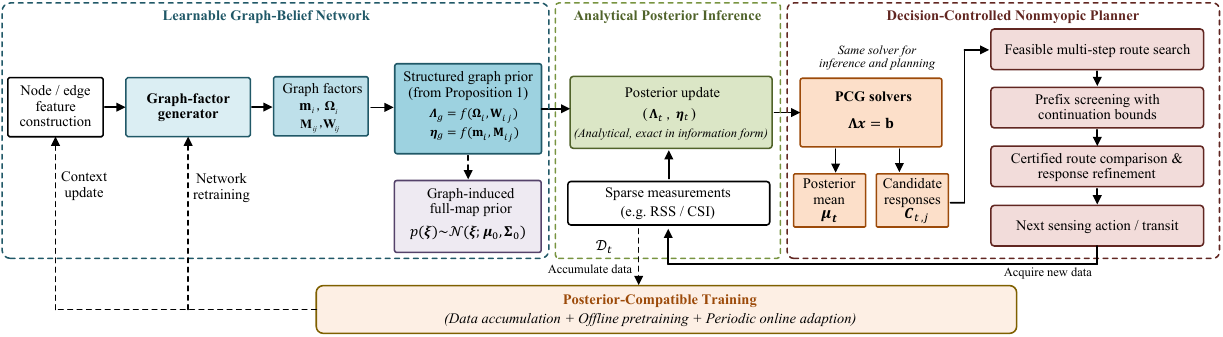}\caption{Overall implementation pipeline of GraphBARC. The graph-factor generator
maps measurement context into interpretable Gaussian node and edge
factors, which are assembled with new measurement information into
the posterior system. Numerical posterior queries provide the map
estimate and candidate responses for the decision-controlled multi-step
planner, forming a closed loop between posterior inference and active
sensing.}
\label{fig:network-architecture}
\end{figure}

After prefix exploration, current incumbent $Q_{t}$ is initialized
as $Q^{0}_{t}$. A candidate replaces it only when its lower bound
is positive. Let $\mathcal{T}_{t}$ contain the unresolved terminal
routes after pruning. Define $g_{t}=\max\{0,\max_{P\in\mathcal{T}_{t}}\overline{\Delta}_{t}(P,Q_{t})\},$
with $g_{t}=0$ when $\mathcal{T}_{t}$ is empty. Thus, $g_{t}$ measures
the largest certified improvement that may still be achieved by an
unresolved route.
\begin{prop}
(Refinement-compatible decision accuracy). Suppose that the certified
comparison and continuation bounds are valid, the prefix search completely
covers the declared route family, and all exclusions and incumbent
replacements are certified. The completed comparison above satisfies
\begin{equation}
0\leq\max_{P\in\mathcal{P}_{t}}J_{t}(P)-J_{t}(Q_{t})\leq g_{t}\label{eq:Refinement-compatible}
\end{equation}
at every completed comparison stage. Previously certified exclusions
remain valid after subsequent response refinement.
\end{prop}
\begin{proof}
See Appendix~\ref{Proof of Proposition 3}.
\end{proof}
This proposition separates route coverage from numerical accuracy.
An excluded branch cannot exceed the incumbent against which its exclusion
was certified, and certified incumbent replacements cannot decrease
the exact model value. Consequently, refining the numerical responses
does not require reopening excluded branches.

Let $\varepsilon_{\mathrm{dec}}\ge0$ be the prescribed tolerance.
If no further certified incumbent improvement is available, planning
stops when $g_{t}\le\varepsilon_{\mathrm{dec}}$. Otherwise, only
the response columns involved in unresolved comparisons are further
refined by the solver. Their residual certificates, route-level covariance,
and affected Gram blocks are updated accordingly, after which the
corresponding decision intervals are recomputed. The prefix search
is not repeated. If the numerical budget is exhausted before the decision
tolerance is met, the planner returns the current incumbent and its
certified decision-gap bound.

\section{Network Implementation and Training}

Realizing the framework in Sections III and IV requires the learned
representation, training objective, and online execution to remain
consistent. The network must generate valid Gaussian factors from
sparse observations, learn how additional measurements should refine
predictions elsewhere, and preserve the fixed-model conditions required
by certified route comparisons. We address these requirements through
an observation-aware graph-factor generator, a disjoint context-assimilation-prediction
training objective, and an online procedure that separates periodic
graph refreshes from analytical measurement updates and numerical
decision refinement.

\subsection{Structured Graph-Factor Generator\label{subsec:Deep-Learning-Graph}}

The network predicts the factors of the joint belief rather than directly
outputting an independent estimate and uncertainty score for each
cell. Node and edge representations determine the factor parameters,
while the assembled precision matrix determines the full-map mean,
covariance, and measurement responses, combing the representation
capability of deep learning with the interpretability of Bayesian
inference.

\begin{figure}[!t]
\centering\includegraphics[width=1.0\linewidth]{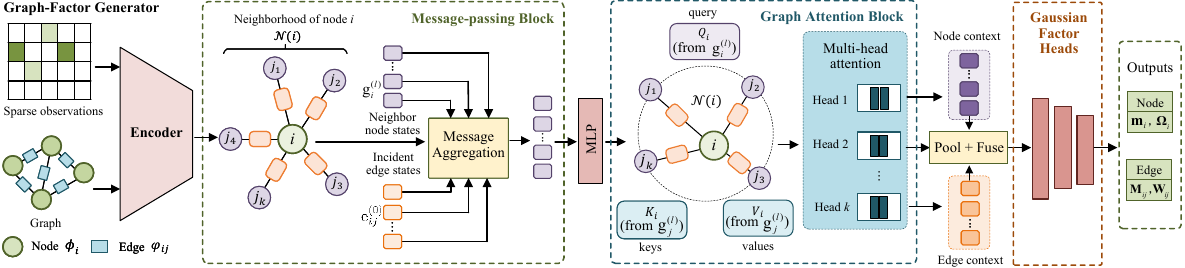}\caption{Observation-aware graph-factor generator. Local message passing and
graph attention aggregate sparse measurement context, and the output
heads produce Gaussian node factors and edge factors, which jointly
define the full-map belief used for analytical assimilation and active
sensing.}
\label{fig:network-architecture-1}
\end{figure}

\subsubsection{Observation-Aware Graph Encoding}

For a context set $\mathcal{D}_{\tau}$, define the input feature
of cell $i$ as $\boldsymbol{\phi}_{i}=[\mathbf{p}^{\mathrm{T}}_{i},(o_{i}\mathbf{y}_{i})^{\mathrm{T}},o_{i}]$,
where $\mathbf{p}_{i}$ is the cell coordinate, $o_{i}=\mathbb{I}\{i\in\mathcal{D}_{\tau}\}$
is the binary observation indicator, and $\mathbf{y}_{i}$ is the
channel observation. For an edge $(i,j)$, the geometric feature is
$\boldsymbol{\varphi}_{ij}=[d_{ij},\varPhi_{ij}]$, where $d_{ij}$
and $\varPhi_{ij}$ are the relative distance and direction between
the two cells. An observation-aware embedding is formed as
\begin{equation}
\tilde{\boldsymbol{\phi}}_{i}=\mathrm{LN}(f_{\mathrm{con}}(\mathbf{p}^{\mathrm{T}}_{i},(o_{i}\mathbf{y}_{i})^{\mathrm{T}})+f_{\mathrm{emb}}(o_{i})),\label{eq:node and edge encoder-1}
\end{equation}
where $f_{\mathrm{con}}(\cdot)$ is the feature encoder, $f_{\mathrm{emb}}(\cdot)$
is the embedding function for the indicator, and $\mathrm{LN}(\cdot)$
denotes layer normalization. The node and edge features are further
encoded by
\begin{equation}
\mathbf{g}^{(0)}_{i}=f_{\mathrm{e1}}(\tilde{\boldsymbol{\phi}}_{i}),\hspace*{0.3cm}\mathbf{e}^{(0)}_{ij}=f_{\mathrm{e2}}(\tilde{\boldsymbol{\phi}}_{i},\tilde{\boldsymbol{\phi}}_{j},\boldsymbol{\varphi}_{ij}).\label{eq:node and edge encoder}
\end{equation}
Here, $f_{\mathrm{e1/e2}}(\cdot)$ denote the node and edge encoders,
and $\mathbf{g}^{(0)}_{i},\mathbf{e}^{(0)}_{ij}\in\mathbb{R}^{B}$
with $B$ being the embedding dimension.

\subsubsection{Graph Context Extraction}

The first stage uses local message passing to encode observation-dependent
relations between neighboring cells. Let $\mathcal{N}(i)=\{j\mid(i,j)\in\mathcal{E}\}$
be the neighborhood of node. At layer $l$,
\begin{equation}
\bar{\mathbf{v}}^{(l)}_{i}=\frac{1}{|\mathcal{N}(i)|}\sum_{j\in\mathcal{N}(i)}f^{(l)}_{\mathrm{v}}(\mathbf{g}^{(l)}_{i},\mathbf{g}^{(l)}_{j},\mathbf{e}^{(0)}_{ij})\label{eq:local message}
\end{equation}
and the node state is updated by
\begin{equation}
\mathbf{g}^{(l+1)}_{i}=\mathrm{LN}(\mathbf{g}^{(l)}_{i}+f^{(l)}_{\mathrm{u}}(\mathbf{g}^{(l)}_{i},\bar{\mathbf{v}}^{(l)}_{i})),\label{eq:node state}
\end{equation}
where $f^{(l)}_{\mathrm{v}}(\cdot)$ and $f^{(\ell)}_{\mathrm{u}}(\cdot)$
are the message-passing and node-update networks, respectively. After
$L_{m}$ layers, the edge representation is refined as $\mathbf{\tilde{e}}_{ij}=f_{\mathrm{r}}([\mathbf{e}^{(0)}_{ij},\mathbf{g}^{(L)}_{i},\mathbf{g}^{(L)}_{j}]).$
Here, $f_{\mathrm{r}}(\cdot)$ is the edge-refinement network.

A graph-attention stage then assigns observation-dependent weights
to neighboring representations. At attention layer $l$, node $i$
attends exclusively to its graph-connected nodes:
\begin{equation}
\mathbf{g}^{(l+1)}_{i}=\mathrm{GraphAtten}^{(l)}(\mathbf{g}^{(l)}_{i},\{\mathbf{g}^{(l)}_{j}|j\in\mathcal{N}(i)\cup\{i\}),\label{eq:observed-node information}
\end{equation}
where $\mathbf{g}^{(l)}_{i}$ provides the query, and the graph-connected
nodes $\{\mathbf{g}^{(l)}_{j}|j\in\mathcal{N}(i)\cup\{i\}$ provide
the keys and values. This attention block learns observation-dependent
importance weights for graph-connected nodes. By stacking multiple
attention layers, measurement information can propagate along multi-hop
graph paths while avoiding unrestricted all-to-all attention.

\subsubsection{Gaussian Factor Heads}

This component converts the learned representations into Gaussian
graph factors. Given the final graph attention output $\tilde{\mathbf{g}}_{i}=\mathbf{g}^{(L)}_{i}$,
the node head generates a node prior and a positive-definite precision
matrix:
\begin{equation}
\mathbf{m}_{i}=f_{m}(\tilde{\mathbf{g}}_{i}),\hspace*{0.3cm}\boldsymbol{\Omega}_{i}=\mathrm{Diag}(\mathrm{softp}(f_{\Omega}(\tilde{\mathbf{g}}_{i})))+\epsilon_{\Omega}\mathbf{I}_{d}.\label{eq:node factors}
\end{equation}
Here, $f_{m}(\cdot)$ and $f_{\Omega}(\cdot)$ are the node-mean and
node-precision heads, respectively, $\epsilon_{\Omega}>0$ is a small
constant. The function $\mathrm{softp}(\boldsymbol{x})=\log(1+\exp(\boldsymbol{x}))$
ensures positive entries, and $\mathrm{Diag}(\cdot)$ forms a diagonal
matrix from a vector.

For each edge $(i,j)$, the transformation head produces
\begin{equation}
\mathbf{M}_{ij}=\mathrm{\mathrm{Diag}}(\exp(\beta_{M}\tanh(f_{M}(\mathbf{z}_{ij})))),\label{eq:edge transform}
\end{equation}
where $\mathbf{z}_{ij}=[\tilde{\mathbf{g}}_{i},\tilde{\mathbf{g}}_{j},\mathbf{\tilde{e}}_{ij}]$,
$\beta_{M}$ bounds the scaling range, and $f_{M}(\cdot)$ is the
transformation head. The edge precision is
\begin{equation}
\mathbf{W}_{ij}=\mathrm{Diag}(\mathrm{softp}(f_{w}(\mathbf{z}_{ij})+\bar{w}\mathrm{tanh}(\mathbf{\alpha}_{ij}))\kappa(d_{ij})+\epsilon_{w}),\label{eq:edge precision}
\end{equation}
where $f_{w}(\cdot)$ is the edge-precision head, $\bar{w}>0$ is
a learnable scale, and $\epsilon_{W}>0$ is a small constant. Here,
$\mathbf{\alpha}_{ij}$ is the attention score from the attention
block, and $\kappa(d_{ij})=\exp(-d_{ij}/\ell_{W})$ is the distance-decay
kernel with decay length $\ell_{W}$.

\subsection{Posterior-Compatible Graph Training}

Training must evaluate both the belief generated from sparse context
and its response to additional evidence. At each training iteration,
the context set is partitioned into three disjoint, nonempty subsets:
$\mathcal{D}_{\tau}=$$\mathcal{D}_{\mathcal{C}}\dot{\cup}\mathcal{D}_{\mathcal{B}}\dot{\cup}\mathcal{D}_{\mathcal{P}}$,
where $\dot{\cup}$ denotes disjoint union. The subsets provide context,
assimilation evidence, and prediction targets, respectively.

The graph generator receives only $\mathcal{D}_{\mathcal{C}}$, and
measurements in $\mathcal{D}_{\mathcal{B}}$ and $\mathcal{D}_{\mathcal{P}}$
are not supplied to the graph-generation forward pass. Let $(\boldsymbol{\hat{\Lambda}}_{g},\boldsymbol{\hat{\eta}}_{g})=f_{graph}(\mathcal{D}_{\mathcal{C}})$
be the natural parameters assembled from the generated factors. The
graph-induced mean and covariance are $\boldsymbol{\hat{\mu}}_{\mathcal{C}}=\boldsymbol{\hat{\Lambda}}^{-1}_{g}\boldsymbol{\hat{\eta}}_{g}$
and $\boldsymbol{\hat{\Sigma}}_{\mathcal{C}}=\boldsymbol{\hat{\Lambda}}^{-1}_{g},$
where the inverse is obtained using differentiable numerical solvers.
Additional evidence from $\mathcal{D}_{\mathcal{B}}$ is assimilated
through the analytic Bayesian update $(\boldsymbol{\hat{\mu}}_{\mathcal{B}},\boldsymbol{\hat{\Sigma}}_{\mathcal{B}})=\mathrm{PosteriorUpdate}(\boldsymbol{\hat{\mu}}_{\mathcal{C}},\boldsymbol{\hat{\Sigma}}_{\mathcal{C}},\mathcal{D}_{\mathcal{B}})$.

Define the robust mean-fitting loss as
\begin{equation}
\mathcal{L}_{H}(\boldsymbol{\mu},\mathcal{D})=\frac{1}{d|\mathcal{D}|}\sum_{i\in\mathcal{D}}\sum^{d}_{r=1}h_{\delta_{H}}(y_{i,r}-\mu_{i,r})
\end{equation}
where $\delta_{H}>0$ is its transition threshold and
\begin{equation}
h_{\delta_{H}}(e)=\begin{cases}
\dfrac{e^{2}}{2\delta_{H}}, & |e|<\delta_{H},\\[1mm]
|e|-\dfrac{\delta_{H}}{2}, & |e|\ge\delta_{H}.
\end{cases}
\end{equation}
To score noisy observations consistently with the observation model,
define $\mathbf{V}_{i}(\boldsymbol{\Sigma})=\mathbf{E}_{i}\boldsymbol{\Sigma}\mathbf{E}^{\mathsf{T}}_{i}+\mathbf{R}_{i}$.
The marginal predictive negative log-likelihood is
\begin{equation}
\adjustbox{max width=.88\linewidth}{\ensuremath{\begin{aligned}\mathcal{L}_{NLL}(\boldsymbol{\mu},\boldsymbol{\Sigma},\mathcal{D}) & =\frac{1}{|\mathcal{\mathcal{D}}|}\sum_{i\in\mathcal{\mathcal{D}}}\frac{1}{2}((\mathbf{y}_{i}-\boldsymbol{\mu}_{i})^{\mathsf{T}}\mathbf{V}_{i}(\boldsymbol{\Sigma})^{-1}(\mathbf{y}_{i}-\boldsymbol{\mu}_{i})\\
 & \hspace*{0.5cm}+\log\mathrm{det}|\mathbf{V}_{i}(\boldsymbol{\Sigma})|+d\log(2\pi)),
\end{aligned}
}}\label{eq:graph_edge-1-1}
\end{equation}

The final training objective is formulated as
\begin{equation}
\begin{aligned}\mathcal{L}_{\mathrm{gra}} & =\frac{1}{2}\Big[\mathcal{L}_{H}(\boldsymbol{\hat{\mu}}_{\mathcal{C}},\mathcal{D}_{\mathcal{B}})+\mathcal{L}_{H}(\boldsymbol{\hat{\mu}}_{\mathcal{B}},\mathcal{\mathcal{D}_{\mathcal{P}}})\Big]\\
 & \hspace*{0.5cm}+\lambda_{b}\Big[\mathcal{L}_{\mathrm{NLL}}(\operatorname{sg}(\boldsymbol{\hat{\mu}}_{\mathcal{C}}),\boldsymbol{\hat{\Sigma}}_{\mathcal{C}},\mathcal{D}_{\mathcal{B}})\\
 & \hspace*{1cm}+\lambda_{p}\mathcal{L}_{\mathrm{NLL}}(\operatorname{sg}(\boldsymbol{\hat{\mu}}_{\mathcal{B}}),\boldsymbol{\hat{\Sigma}}_{\mathcal{B}},\mathcal{\mathcal{D}_{\mathcal{P}}})\Big].
\end{aligned}
\label{eq:negative log-likelihood function}
\end{equation}
The coefficients $\lambda_{b},\lambda_{p}\ge0$ control the regularization
terms, and the operator $\operatorname{sg}(\cdot)$ stops gradients.
The first term trains the graph model to produce a reliable belief
from sparse context alone, while the second encourages the model to
support analytical posterior refinement with additional measurements.
The network is trained to generate graph factors that are compatible
with both sparse prediction and measurement propagation.

\subsection{Implementation}

GraphBARC represents the current belief by its natural parameters
and computes only the inverse actions required for reconstruction,
training, and planning. It computes the mean $\boldsymbol{\mu}_{t}$
and candidate responses $\mathbf{C}_{t,j}$ through $\boldsymbol{\Lambda}_{t}\boldsymbol{\mu}_{t}=\boldsymbol{\eta}_{t}$
and $\boldsymbol{\Lambda}_{t}\mathbf{C}_{t,j}=\mathbf{E}^{\mathsf{T}}_{j}$,
respectively. Both queries use matrix-free preconditioned conjugate
gradient (PCG) iterations with a symmetric positive-definite preconditioner.
They share the same numerical solver, differing only in their right-hand
sides.

Within each planning call, the belief, task weights, noise model,
and candidate domain remain fixed. Shared responses and Gram blocks
support the certified prefix screening and selective numerical refinement.
The planner returns a feasible route and its attained decision gap
$g_{t}$, distinguishing satisfaction of the prescribed tolerance,
$g_{t}\le\varepsilon_{\mathrm{dec}}$, from termination at a numerical
work limit. The selected route is executed as a batch, while measurements
are assimilated sequentially.

The network is initialized using the sparse observations $\mathcal{D}_{0}$
and adapted periodically using the accumulated dataset $\mathcal{D}_{t}$.
Between graph refreshes, its factors remain fixed and only newly acquired
measurements are assimilated. At a scheduled refresh, the regenerated
belief replaces the current information system. The final map is obtained
by solving the mean system after the last update. The complete procedure
is summarized in Algorithm~\ref{Alg1:online_training_active_sampling}.

\begin{algorithm}[!t]
\begin{algorithmic}[1]

\Require Initial observation set $\mathcal{D}_{0}$, sensing budget
$B=(\varGamma,\ell_{\max})$, lookahead horizon $K$, and tolerance
$\varepsilon_{\mathrm{dec}}$

\Ensure Radio map construction and sampling trajectory

\State Pretrain $f_{graph}(\cdot)$ with $\mathcal{D}_{0}$; assemble
$\boldsymbol{\Lambda}_{0}$, $\boldsymbol{\eta}_{0}$

\State Initialize $t\leftarrow0$, $\ell_{t}\leftarrow\ell_{\max}$,
$n\leftarrow0$ and $\mathbf{p}^{\mathrm{u}}_{t}$

\While{under the sensing budget $B$}

\State Set $\mathit{\Gamma}_{t}=\min\{K,\varGamma-n\}$; obtain $\boldsymbol{\hat{\mu}}_{t}$

\State Compute $\widehat{\mathbf{C}}_{t,j}$ and $\widehat{\mathbf{G}}_{t,ij}$
for all candidates

\State Construct a greedy route $Q^{0}_{t}$; initialize $\mathcal{T}_{t}=\{Q^{0}_{t}\}$

\State Generate feasible prefixes $P=(j_{1},...,j_{k})$ for $k\leq\mathit{\Gamma}_{t}$;
prune those with $U_{t}(P;Q^{0}_{t})<0$; add $P$ to $\mathcal{T}_{t}$
if $k=\mathit{\Gamma}_{t}$

\State Set $Q_{t}\leftarrow Q^{0}_{t}$

\While{true}

\State For all $P\in\mathcal{T}_{t}$, compute $\underline{\Delta}_{t}(P,Q_{t})$
and $\overline{\Delta}_{t}(P,Q_{t})$

\State Remove routes satisfying $\ensuremath{\overline{\Delta}_{t}(P,Q_{t})\le0}$

\If{$\ensuremath{\underline{\Delta}_{t}(P,Q_{t})>0}$}

\State Set $Q_{t}\leftarrow P$, remove $P$ from $\mathcal{T}_{t}$;
continue

\EndIf

\State Compute $g_{t}$, and take $g_{t}=0$ when $\mathcal{T}_{t}=\varnothing$

\State If $g_{t}\le\varepsilon_{\mathrm{dec}}$ or the numerical
limit is reached, report $g_{t}$ and break

\State Refine response columns, $\underline{\Delta}_{t}(P,Q_{t})$
and $\overline{\Delta}_{t}(P,Q_{t})$ for $P\in\mathcal{T}_{t}$ under
higher-precision PCG

\EndWhile

\State Execute the route $Q_{t}$; update $\mathcal{D}_{t+1}$, $\ell_{t+1}$,
$\mathbf{p}^{\mathrm{u}}_{t+1}$

\State Refresh $\boldsymbol{\Lambda}_{t+1},\boldsymbol{\eta}_{t+1}$
with analytical Bayesian updates

\If {periodic adaptation is triggered}

\State Train $f_{graph}(\cdot)$ with $\mathcal{D}_{t+1}$; regenerate
$\boldsymbol{\Lambda}_{t+1}$, $\boldsymbol{\eta}_{t+1}$

\EndIf

\State Set $t\leftarrow t+1$, $n\leftarrow n+\mathit{\Gamma}_{t}$

\EndWhile

\State Output $\boldsymbol{\hat{\mu}}_{t}$ and $\mathcal{D}_{t}$

\end{algorithmic}

\caption{Active Sensing and Inference with GraphBARC.}

\label{Alg1:online_training_active_sampling}
\end{algorithm}

\section{Experiment Results}

This section evaluates the proposed GraphBARC in terms of reconstruction
accuracy and sensing efficiency. The experiments are conducted using
channel data generated by the Sionna ray-tracing (RT) simulator \citep{sionnarttechnicalreport}.

In the simulation, we consider an urban environment covering an area
of 860~m $\times$ 680~m, where seven transmitters (TXs) are randomly
deployed at heights ranging from 25 to 40 m. Each TX is equipped with
a 16-element antenna array. The simulations use a band-limited signal
with a carrier frequency of 3.5 GHz and an effective bandwidth of
100 MHz. The ground truth radio map is computed on a $128\times128$
planar grid at a fixed altitude of 50 m. In total, 10 radio maps are
generated under different TX location configurations for evaluation.

The proposed GraphBARC is compared with the following representative
baselines:
\begin{enumerate}
\item BAL-EGP \citep{PolySadYe:J24}: It uses an ensemble of Gaussian processes
to reconstruct the radio map. The next location is selected by maximizing
the ensemble predictive variance, and the model posterior is updated
with new samples.
\item DRUE \citep{ShreRomeChe:J23}: It uses two autoencoders to estimate
the radio map and uncertainty map. During active sensing, high-uncertainty
locations are selected as candidate goals, and A{*} search is used
to plan grid paths toward these goals.
\item Flow Matching \citep{SunShicXiang:C26}: It uses a pretrained flow-matching
UNet for radio map reconstruction. The variance of multiple stochastic
radio map reconstructions gives the uncertainty map, which guides
active sampling.
\end{enumerate}

Let $\mathbf{\hat{X}}$ be the reconstructed radio map (in dB scale).
We evaluate the reconstruction performance using the mean absolute
error (MAE) and RMSE, which are defined as $\mathrm{MAE}=\frac{1}{M}||\mathbf{X}-\mathbf{\hat{X}}||_{1}$
and $\mathrm{RMSE}=\frac{1}{\sqrt{M}}||\mathbf{X}-\mathbf{\hat{X}}||_{2}$.
To reflect the energy error, selected experiments additionally evaluate
the normalized mean-squared error (NMSE) in the linear power domain,
given by $\mathrm{NMSE}=||10^{\mathbf{X}/10}-10^{\mathbf{\hat{X}}/10}||^{2}_{2}/||10^{\mathbf{X}/10}||^{2}_{2}$.

Specifically, 3\% of the grid cells are randomly observed for model
initialization. All trajectories start from the origin, and each sampling
action is restricted to three hops on the eight-connected grid. We
set the lookahead horizon $K=4$ and $\varepsilon_{\mathrm{dec}}=10^{-6}$.
Candidate responses initially use 37 PCG iterations, with refinements
of 16 iterations and a maximum of 60. Retraining is performed after
every 100 newly collected measurements. Sampling stops when reaching
the 15\% grid-cell observation budget, including the initial 3\%.

\subsection{Active Radio Map Reconstruction}

\begin{figure}[!t]
\centering \includegraphics[width=1.0\linewidth]{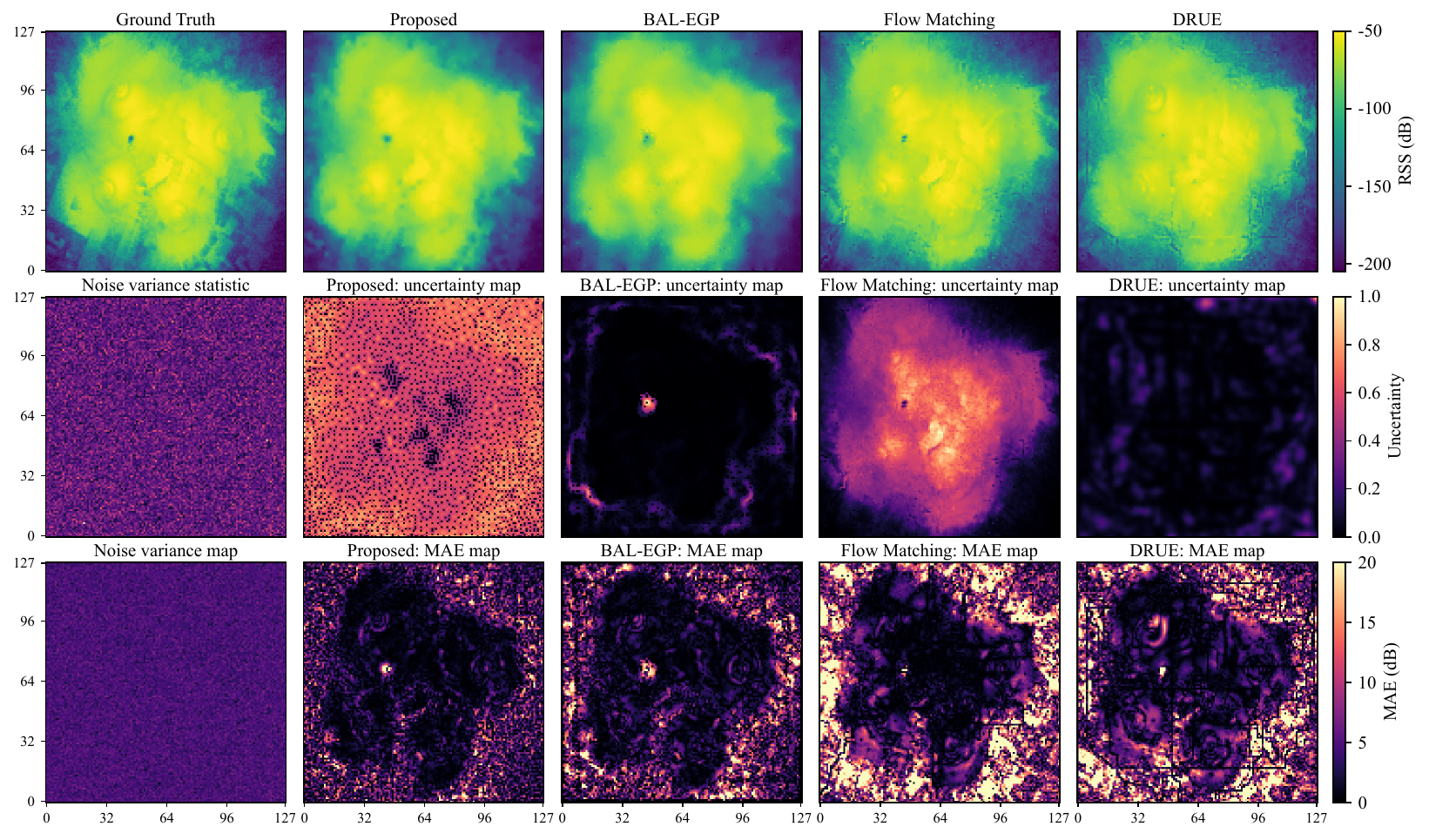}
\caption{Visual comparison of reconstructed radio maps, normalized uncertainty
maps, and spatial MAE distributions under 15\% observation coverage.}
\label{fig:radio_map_example}
\end{figure}

We first evaluate the radio-map reconstruction accuracy and the benefits
of nonmyopic sensing.

Table~\ref{Tab:radio map construction accuracy} compares reconstruction
accuracy on the common set of cells that remain unobserved by all
methods. For GraphBARC, $K=1$ reduces to the one-step greedy policy
in (\ref{eq:one-step Bayes risk}). It is observed that GraphBARC
with $K=4$ outperforms all baseline methods, achieving the lowest
errors across all metrics. Compared with BAL-EGP, it reduces these
metrics by $16.2\%$, $16.1\%$, and $41.5\%$, respectively, while
its MAE and RMSE improvements over DRUE and Flow Matching exceed $52\%$.
Although Flow Matching exhibits relatively large MAE and RMSE, it
achieves the lowest NMSE among the external baselines. Notably, increasing
the planning horizon from $K=1$ to $K=4$ further reduces MAE and
RMSE by $7.1\%$ and $7.3\%$, respectively, demonstrating the advantage
of multi-step sensing over one-step greedy acquisition.

Table~\ref{Tab:Observation demand} reports the number of observations
required to reach the prescribed full-map reconstruction accuracy.
For the RMSE target, GraphBARC with $K=4$ requires only 1,760 observations,
$4.3\%$ fewer than $K=1$ and $39.5\%$ fewer than BAL-EGP. For the
NMSE target, the corresponding reductions increase to $19.4\%$ and
$79.4\%$, respectively. Flow Matching and DRUE require 1.87 and 4.66
times as many observations, respectively. These results further demonstrate
that joint multi-step sensing can improve measurement efficiency under
prescribed reconstruction-error constraints.

Fig.~\ref{fig:radio_map_example} visualizes the reconstructed radio
maps, normalized uncertainty maps, and spatial error distributions.
GraphBARC captures the main propagation patterns and spatial variations,
although some fine-scale structures remain smoothed. BAL-EGP exhibits
more pronounced smoothing, whereas Flow Matching and DRUE recover
the central high-RSS regions but incur larger errors in several peripheral
and transition areas. The uncertainty maps reveal different spatial
characteristics. Flow Matching assigns high uncertainty to the central
high-power region, where the actual errors are comparatively small,
while BAL-EGP and DRUE exhibit localized uncertainty patterns that
do not align with their error distributions. The proposed uncertainty
map shows partial spatial correspondence with reconstruction errors.
Although this correspondence is not exact, it suggests that the estimated
uncertainty captures some spatial variation in reconstruction difficulty.

\begin{table}
\caption{Reconstruction accuracy on common unobserved cells}

\centering\label{Tab:radio map construction accuracy} \renewcommand\arraystretch{1.3}
\begin{tabular}{ p{3.0cm}<{\centering}| p{2.5cm}<{\centering}| p{2.5cm}<{\centering}| p{2.5cm}<{\centering}}
\hline

Scheme          & MAE (dB)      & RMSE (dB)    & NMSE \\

\hline
Flow Matching    &$7.845$   &$11.475$	&$0.162$\\
\hline

DRUE             &$7.726$   &$11.110$   &$0.292$\\
\hline

BAL-EGP          &$4.307$   &$6.311$    &$0.176$\\
\hline

Proposed ($K=1$)         &$3.886$ &$5.711$ &$0.107$ \\
\hline

Proposed ($K=4$)         &$\mathbf{3.609}$ &$\mathbf{5.293}$ &$\mathbf{0.103}$ \\
\hline

\end{tabular}
\end{table}

Fig.~\ref{fig:SampleVsMthodRMSE} compares RMSE across observation
coverage levels on a representative radio map, with all methods evaluated
on a common set of unobserved cells. Both proposed variants consistently
outperform the external baselines, with the planner $K=4$ achieving
the lowest RMSE across all evaluated observation ratios. Compared
with BAL-EGP, its RMSE reduction increases overall from $11.0\%$
at $10\%$ observation coverage to $29.8\%$ at $35\%$, while the
improvements over DRUE and Flow Matching reach $60.7\%$ and $62.3\%$,
respectively. The proposed $K=4$ consistently improves upon $K=1$
by jointly evaluating gains over feasible multi-step trajectories,
with relative improvements of $6.3\%$ and $5.3\%$ at $10\%$ and
$35\%$ observation coverage, respectively. These results consistently
support the benefit of nonmyopic sensing over greedy selection.

\begin{table}
\caption{Observations required to meet target reconstruction errors}

\centering\label{Tab:Observation demand} \renewcommand\arraystretch{1.3}
\begin{tabular}{p{3.0cm}<{\centering}|p{1.6cm}<{\centering}|p{1.6cm}<{\centering}|p{1.6cm}<{\centering}|p{1.6cm}<{\centering}}
\hline

\multirow{2}{*}{Scheme} &
\multicolumn{2}{c|}{RMSE $\leq$ 5.5 dB} &
\multicolumn{2}{c}{NMSE $\leq$ 0.11} \\
\cline{2-5}

     & Num. & Relative &Num. & Relative\\
\hline
Flow Matching & $13{,}504$ & $7.67\times$ & $1{,}708$ & $1.87\times$\\
\hline
DRUE          & $11{,}140$ & $6.33\times$ & $4{,}248$ & $4.66\times$\\
\hline
BAL-EGP       & $2{,}908$  & $1.65\times$ & $4{,}428$ & $4.85\times$\\
\hline

Proposed ($K=1$)      &$1{,}840$ &$1.05\times$ &$1{,}132$ &$1.24\times$\\
\hline

Proposed ($K=4$)      &$\mathbf{1{,}760}$ &$\mathbf{1.00}\times$ &$\mathbf{912}$ &$\mathbf{1.00}\times$\\
\hline
\end{tabular}
\end{table}

\begin{figure}[!t]
\centering
\includegraphics[width=0.65\linewidth]{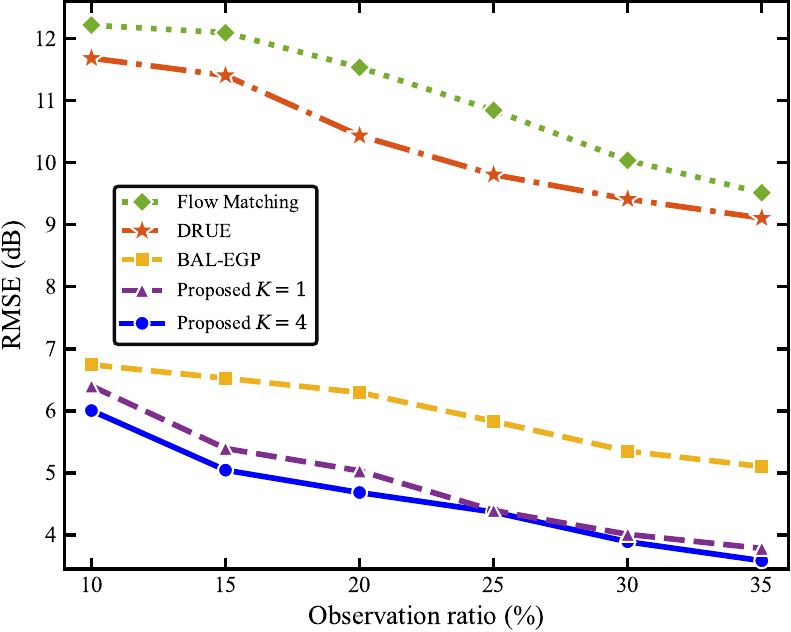} \caption{RMSE versus observation coverage on a representative radio map.}
\label{fig:SampleVsMthodRMSE}
\end{figure}

\subsection{Evaluation of Active Sampling Quality}

This experiment evaluates whether GraphBARC selects informative measurements
beyond its own reconstruction model. To isolate sampling quality from
the reconstruction backend, we reconstruct radio maps using three
independent common estimators: $k$-nearest-neighbor (KNN) interpolation,
radial basis function (RBF) interpolation, and Gaussian process regression
with a Matérn covariance kernel (Matérn-GPR) \citep{rasmussen2006gaussian}.
All estimators use the measurements collected by each sensing strategy,
including its initial observations. We use four neighbors for KNN,
local linear RBF interpolation with 128 neighbors, and a Matérn-$3/2$
kernel for GPR.

\begin{figure}[!t]
\centering
\includegraphics[width=0.75\linewidth]{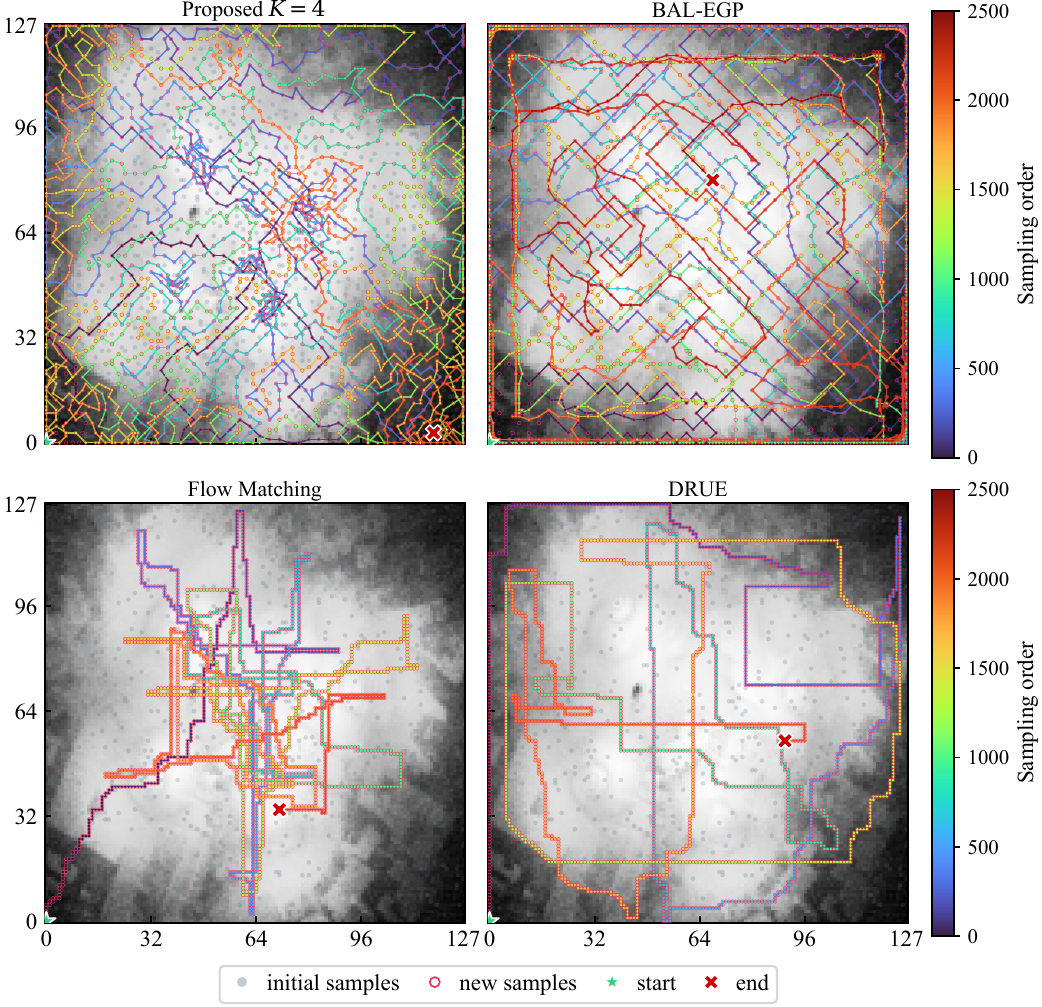}
\caption{Sampling trajectories generated by different active strategies.}
\label{fig:trajectories}
\end{figure}

\begin{table}
\caption{RMSE of independent reconstruction methods using measurements acquired
by different sensing strategies}

\centering\label{Tab:radio map different methods} \renewcommand\arraystretch{1.5}
\begin{tabular}{ p{3.0cm}<{\centering}| p{2.5cm}<{\centering}| p{2.5cm}<{\centering}| p{2.5cm}<{\centering}}
\hline

Scheme          & KNN      & RBF    & Matérn-GPR \\

\hline

Random          &$6.201$   &$5.529$	&$5.675$\\
\hline

Flow Matching    &$8.965$   &$7.225$	&$7.869$\\
\hline

DRUE             &$7.975$   &$6.269$   &$6.657$\\
\hline

BAL-EGP          &$6.566$   &$5.696$    &$5.841$\\
\hline

Proposed ($K=1$)  &$5.979$ &$5.520$ &$5.735$ \\
\hline

Proposed ($K=4$)  &$\mathbf{5.470}$ &$\mathbf{5.233}$ &$\mathbf{5.332}$ \\
\hline

\end{tabular}
\end{table}

Fig.~\ref{fig:trajectories} first visualizes the sampling trajectories
generated by different methods. GraphBARC with $K=4$ distributes
measurements across broad spatial coverage while prioritizing propagation-critical
regions, such as shadow areas and LOS/NLOS boundaries, where spatial
heterogeneity often leads to higher posterior uncertainty. BAL-EGP
also provides broad spatial coverage, but adapts less to propagation
structures. Flow Matching concentrates mainly on the central high-RSS
region, resulting in weaker reconstruction in low-RSS and boundary
areas, while DRUE emphasizes uncertainty boundaries but leaves larger
regions underexplored. These differences illustrate how acquisition
strategies influence the spatial distribution of collected measurements.

Table~\ref{Tab:radio map different methods} compares the resulting
RMSE across the three independent reconstruction methods. It is observed
that GraphBARC with $K=4$ achieves the lowest RMSE under all three
estimators. Compared with random sampling, it reduces RMSE by $11.8\%$,
$5.4\%$, and $6.0\%$ under KNN, RBF, and Matérn-GPR, respectively.
The corresponding reductions over BAL-EGP are $16.7\%$, $8.1\%$,
and $8.7\%$. Notably, $K=4$ further improves upon $K=1$ by $8.5\%$,
$5.2\%$, and $7.0\%$, respectively. Flow Matching and DRUE may underperform
random sampling because their trajectories are spatially biased, whereas
random sampling provides more balanced global coverage. Overall, the
consistent improvements across different estimators indicate that
the proposed nonmyopic strategy acquires more informative measurements
whose reconstruction benefits extend beyond the graph-based model.

\subsection{Reliability and Efficiency}

This section evaluates reliability and efficiency of the proposed
GraphBARC. All experiments are conducted on a workstation equipped
with an Intel Core i7-10700K CPU, 32 GB RAM, and an NVIDIA GeForce
RTX 4060 Ti GPU.

Table~\ref{Tab:different initialization} evaluates the reliability
of GraphBARC with initialization ratios of $1\%$\textendash $4\%$
of the grids, followed by active sensing until the total observation
coverage reaches $10\%$ or $15\%$. GraphBARC consistently outperforms
random sampling across all initialization settings, reducing RMSE
by $3.5\%$\textendash $13.2\%$ at $10\%$ coverage and $10.3\%$\textendash $12.9\%$
at $15\%$ coverage. Notably, the reconstruction performance exhibits
relatively low sensitivity to initialization size at $15\%$ coverage,
with RMSE ranging from $4.774$ to $4.916$ dB across all tested settings.
Even with only $1\%$ initial observations, the proposed method achieves
RMSEs of $5.704$ and $4.774$ dB at $10\%$ and $15\%$ coverage,
respectively, outperforming random sampling by $13.2\%$ and $12.9\%$.
These results indicate that GraphBARC can be effectively initialized
with limited measurements and progressively improve reconstruction
through active sensing, without requiring an extensive initial observation
set.

Fig.~\ref{fig:Validity and tightness} validates the numerical decision
bounds using 50 frozen belief states from all maps. In Fig.~\ref{fig:Validity and tightness}(a),
the decision regret, evaluated against an independent high-accuracy
reference optimum, is covered by the reported upper bound in all states
after accounting for reference uncertainty and rounding allowance.
All states satisfy the prescribed tolerance, and 47 exhibit regret
within the rounding allowance. Fig.~\ref{fig:Validity and tightness}(b)
isolates the benefit of retaining shared response errors when comparing
competing routes. The mean error radius decreases by $8.81\%$, $19.80\%$,
and $33.42\%$ relative to separate route-error bounds when the routes
share one, two, and three measurement locations, respectively. These
results support the validity of the stopping criterion and demonstrate
that shared-error comparisons provide tighter bounds under the fixed
learned belief and declared feasible-route family.

\begin{table}
\caption{Sensitivity of reconstruction to the initialization ratios}

\centering\label{Tab:different initialization} \renewcommand\arraystretch{1.3}
\begin{tabular}{ p{2.8cm}<{\centering}| p{1.8cm}<{\centering}| p{1.8cm}<{\centering}| p{1.8cm}<{\centering}| p{1.8cm}<{\centering}}
\hline

\multirow{2}{*}{Scheme} &
\multicolumn{2}{c|}{10\%} &
\multicolumn{2}{c}{15\%} \\
\cline{2-5}

          & MAE      & RMSE    & MAE      & RMSE \\

\hline
Random                &4.326   &6.569    &3.559  &5.481\\
\hline
1\% + Active          &3.742	&5.704	&3.160	&4.774\\
\hline
2\% + Active          &4.019	&6.217	&3.162	&4.891\\
\hline
3\% + Active          &4.073	&6.336	&3.209	&4.846\\
\hline
4\% + Active          &3.879	&5.986	&3.194	&4.916\\
\hline

\end{tabular}
\end{table}

\begin{figure}[!t]
\centering
\includegraphics[width=0.90\linewidth]{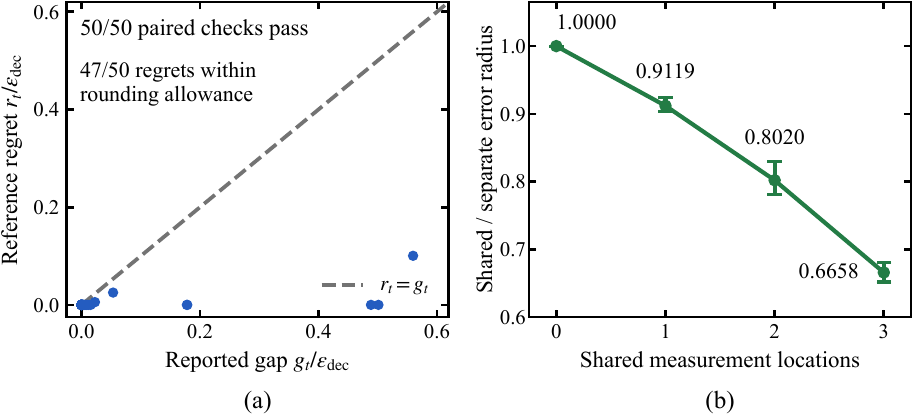} \caption{Validity and tightness of decision-error bounds: a) Reference decision
regret $r_{t}$ versus the reported decision-gap bound $g_{t}$; b)
Shared-to-separate error-radius ratio versus the number of shared
measurement locations.}
\label{fig:Validity and tightness}
\end{figure}

Table~\ref{Tab:Runtime} evaluates the computational efficiency of
four-step planning over 50 frozen belief states using total runtime,
cumulative PCG iterations, generated prefixes, and route-value comparisons.
Fixed-Gram uses fixed-accuracy response solves without prefix pruning;
Adaptive-Gram adds selective PCG refinement; Pruned-Gram additionally
applies error-aware prefix screening; and Shared-Pruned further exploits
shared-error bounds for comparisons. Compared with Fixed-Gram, Adaptive-Gram
reduces PCG iterations by $38.2\%$ and runtime by $13.7\%$, with
essentially unchanged search workload. Adding prefix screening further
reduces generated prefixes and route comparisons by $96.3\%$ and
$97.6\%$, respectively, yielding a $53.3\%$ reduction in runtime.
These results identify prefix pruning as the primary source of search
acceleration. Shared-Pruned slightly reduces the search workload further,
but incurs approximately $7\%$ higher runtime than Pruned-Gram due
to the additional bound-evaluation overhead. Nevertheless, it achieves
an overall $56.8\%$ reduction in planning time relative to Fixed-Gram,
demonstrating the complementary benefits of adaptive response computation
and bound-guided pruning.

\section{Conclusion}

\begin{table}
\caption{Runtime and computational workload of four-step planning over 50 frozen
states}

\centering\label{Tab:Runtime} \renewcommand\arraystretch{1.3}
\begin{tabular}{ p{2.8cm}<{\centering}| p{1.8cm}<{\centering}| p{1.8cm}<{\centering}| p{1.8cm}<{\centering}| p{1.8cm}<{\centering}}
\hline

\multirow{3}{*}{Scheme}     & \multirow{3}{*}{Time (s)}
& \multirow{3}{1.8cm}{\arraybackslash\centering PCG iterations ($\times 10^{-3}$)}
& \multirow{3}{1.8cm}{\arraybackslash\centering Generated prefixes ($\times 10^{-3}$)}
& \multirow{3}{1.8cm}{\arraybackslash\centering Route comp. ($\times 10^{-3}$)} \\

& & & &\\
& & & &\\

\hline
Fixed-Gram             &3.344   &23.33    &1558.3    &672.7\\
\hline
Adaptive-Gram          &2.886	&14.42	&1558.3	&673.1\\
\hline
Pruned-Gram            &1.348	&14.42	&58.34	&15.98\\
\hline
Shared-Pruned          &1.443	&14.42	&58.24	&15.90\\
\hline

\end{tabular}
\end{table}

This paper presented GraphBARC, a graph-structured Bayesian framework
for nonmyopic active radio map construction with certified decision-error
control. Its learned graph-induced belief captures spatial dependencies,
enabling analytical posterior updates and information propagation
across the map. By evaluating joint map-wide Bayes risk reduction
over feasible routes, GraphBARC accounts for measurement redundancy
and future sensing opportunities beyond one-step gains. Shared posterior
responses, residual-certified route comparisons, and safe prefix screening
reduce redundant computation while focusing numerical refinement on
unresolved decisions. Under fixed-belief and complete coverage of
the declared route family, the planner provides a computable decision-gap
bound, with certified exclusions preserved during refinement. Experiments
show that GraphBARC reduces mean RMSE by over $16\%$ relative to
the evaluated methods cuts core planning time by $56\%$ compared
with the fixed-accuracy, unpruned implementation. Its selected measurements
also achieve the lowest RMSE among the evaluated strategies under
external independent reconstruction methods.

Future work will extend graph-structured beliefs to dynamic spatio-temporal
radio environments and generalize the planner to cooperative multi-UAV
sensing. Another promising direction is to learn transferable graph
priors for previously unseen environments, where newly acquired observations
progressively update the posterior belief and refine graph factors
for environment-specific adaptation.

\appendix
\section{Proofs and Auxiliary Bounds}

\subsection{Proof of Proposition 1}\label{app:Proposition 1 proof}

The logarithm of the graph belief in (\ref{eq:graph_prior}) can be
written as
\begin{equation}
\begin{aligned}\log q_{t}(\boldsymbol{\xi};\mathcal{D}_{t}) & =-\frac{1}{2}\sum^{M}_{i=1}(\mathbf{E}_{i}\boldsymbol{\xi}-\mathbf{m}_{i})^{\mathsf{T}}\mathbf{\Omega}_{i}(\mathbf{E}_{i}\boldsymbol{\xi}-\mathbf{m}_{i})\\
 & \hspace*{0.4cm}-\frac{1}{2}\sum_{(i,j)\in\mathcal{E}}(\mathbf{A}_{ij}\boldsymbol{\xi})^{\mathsf{T}}\mathbf{W}_{ij}(\mathbf{A}_{ij}\boldsymbol{\xi})+C,
\end{aligned}
\label{eq:Proof-1}
\end{equation}
where $C$ is independent of $\boldsymbol{\xi}$. Expanding the node
terms and collecting the quadratic and linear terms give
\begin{equation}
\begin{aligned}\log q_{t}(\boldsymbol{\xi};\mathcal{D}_{t}) & =-\frac{1}{2}\boldsymbol{\xi}^{\mathsf{T}}(\sum^{M}_{i=1}\mathbf{E}^{\mathsf{T}}_{i}\mathbf{\Omega}_{i}\mathbf{E}_{i}+\sum_{(i,j)\in\mathcal{E}}\mathbf{A}^{\mathsf{T}}_{ij}\mathbf{W}_{ij}\mathbf{A}_{ij})\boldsymbol{\xi}\\
 & \hspace*{0.4cm}+(\sum^{M}_{i=1}\mathbf{E}^{\mathsf{T}}_{i}\mathbf{\Omega}_{i}\mathbf{m}_{i})^{\mathsf{T}}\boldsymbol{\xi}+C',\\
 & =-\frac{1}{2}\boldsymbol{\xi}^{\mathsf{T}}\boldsymbol{\Lambda}_{t}\boldsymbol{\xi}+\boldsymbol{\eta}^{\mathsf{T}}_{t}\boldsymbol{\xi}+C'
\end{aligned}
\label{eq:eq:Proof-2}
\end{equation}
where $C'$ absorbs all terms independent of $\boldsymbol{\xi}$.

We next establish positive definiteness. For any nonzero vector $\mathbf{z}\in\mathbb{R}^{Md}$,
define $\mathbf{z}_{i}=\mathbf{E}_{i}\mathbf{z}$. Then,
\begin{equation}
\begin{aligned}\mathbf{z}^{\mathsf{T}}\boldsymbol{\Lambda}_{t}\mathbf{z} & =\sum_{(i,j)\in\mathcal{E}}(\mathbf{z}_{i}-\mathbf{M}_{ij}\mathbf{z}_{j})^{\mathsf{T}}\mathbf{W}_{ij}(\mathbf{z}_{i}-\mathbf{M}_{ij}\mathbf{z}_{j})\\
 & \hspace*{0.4cm}+\sum^{M}_{i=1}\mathbf{z}^{\mathsf{T}}_{i}\mathbf{\Omega}_{i}\mathbf{z}_{i}.
\end{aligned}
\label{eq:eq:Proof-2-1}
\end{equation}
Because $\mathbf{z}\neq\mathbf{0}$, at least one block $\mathbf{z}_{i}$
is nonzero. The second sum in (\ref{eq:eq:Proof-2-1}) is strictly
positive under $\mathbf{\Omega}_{i}\succ\mathbf{0}$, while the second
sum is nonnegative under $\mathbf{W}_{ij}\succeq\mathbf{0}$. Hence,
$\mathbf{z}^{\mathsf{T}}\boldsymbol{\Lambda}_{t}\mathbf{z}>0$ for
every $\mathbf{z}\neq\mathbf{0}$, which proves that $\boldsymbol{\Lambda}_{t}\succ\mathbf{0}$.

Finally, let $\boldsymbol{\mu}_{t}=\boldsymbol{\Lambda}^{-1}_{t}\boldsymbol{\eta}_{t}$
and completing the square gives
\begin{equation}
-\frac{1}{2}\boldsymbol{\xi}^{\mathsf{T}}\boldsymbol{\Lambda}_{t}\boldsymbol{\xi}+\boldsymbol{\eta}^{\mathsf{T}}_{t}\boldsymbol{\xi}=-\frac{1}{2}(\boldsymbol{\xi}-\boldsymbol{\mu}_{t})^{\mathsf{T}}\boldsymbol{\Lambda}_{t}(\boldsymbol{\xi}-\boldsymbol{\mu}_{t})+C'',\label{eq:eq:Proof-2-1-1}
\end{equation}
where $C''=\frac{1}{2}\boldsymbol{\eta}^{\mathsf{T}}_{t}\boldsymbol{\Lambda}^{-1}_{t}\boldsymbol{\eta}_{t}$
is independent of $\boldsymbol{\xi}$. Since $\boldsymbol{\Lambda}_{t}\succ\mathbf{0}$,
the normalized density is Gaussian with mean $\boldsymbol{\mu}_{t}=\boldsymbol{\Lambda}^{-1}_{t}\boldsymbol{\eta}_{t}$
and covariance $\boldsymbol{\Sigma}_{t}=\boldsymbol{\Lambda}^{-1}_{t}$,
thereby completing the proof.

\subsection{Proof of Proposition 2}\label{app:Proposition 2 proof}

Let $\mathbf{D}_{t,j}=\mathbf{C}_{t,j}-\widehat{\mathbf{C}}_{t,j}$
and let $\mathbf{D}^{P}_{t}$ stack the corresponding blocks. Since
$\boldsymbol{\Lambda}_{t}\mathbf{D}_{t,j}=\mathbf{E}^{\mathsf{T}}_{j}-\boldsymbol{\Lambda}_{t}\widehat{\mathbf{C}}_{t,j},$
the lower bound gives
\begin{equation}
\|\mathbf{D}_{t,j}\|_{F}\leq\|\boldsymbol{\Lambda}^{-1}_{t}\|_{2}\|\mathbf{E}^{\mathsf{T}}_{j}-\boldsymbol{\Lambda}_{t}\widehat{\mathbf{C}}_{t,j}\|_{F}\leq\epsilon_{t,j}.
\end{equation}
Block stacking yields $\|\mathbf{D}^{P}_{t}\|_{F}\leq\epsilon_{t}(P)$.
Because route locations are distinct, $\|\mathbf{E}_{P}\|_{2}=1$.
Hence,
\begin{equation}
\begin{aligned}\|\mathbf{S}^{P}_{t}-\widehat{\mathbf{S}}^{P}_{t}\|_{2} & =\|\mathrm{sym}(\mathbf{E}_{P}\mathbf{D}^{P}_{t})-\gamma_{t}\mathbf{I}\|_{2}\leq\epsilon_{t}(P)+\gamma_{t}.\end{aligned}
\label{eq: inverse2}
\end{equation}

The stabilization rule ensures $\widehat{\mathbf{S}}^{P}_{t}\succeq\mathbf{R}_{P}\succ\mathbf{0}.$
A valid explicit choice of spectral lower bounds is
\begin{equation}
\begin{aligned}\widehat{s}_{P} & =\max\{\lambda_{\min}(\mathbf{R}_{P_{t}}),\frac{1}{||(\widehat{\mathbf{S}}^{P_{t}}_{t})^{-1}||_{\infty}}\},\\
s_{P} & =\max\{\lambda_{\min}(\mathbf{R}_{P_{t}}),\widehat{s}_{P}-\epsilon_{t}(P_{t})-\gamma_{t}\}.
\end{aligned}
\end{equation}
Here, $\|\cdot\|_{\infty}$ denotes the maximum absolute row sum.

For the positive-definite numerical matrix,
\begin{equation}
\lambda_{\min}(\widehat{\mathbf{S}}^{P}_{t})=\frac{1}{\lambda_{\max}((\widehat{\mathbf{S}}^{P}_{t})^{-1})}\geq\frac{1}{\|(\widehat{\mathbf{S}}^{P}_{t})^{-1}\|_{\infty}}.
\end{equation}
Together with $\widehat{\mathbf{S}}^{P}_{t}\succeq\mathbf{R}_{P}$,
this establishes $\widehat{s}_{P}$. It implies
\begin{equation}
\lambda_{\min}(\mathbf{S}^{P}_{t})\geq\widehat{s}_{P}-\epsilon_{t}(P)-\gamma_{t}.
\end{equation}
Combining it with $\mathbf{S}^{P}_{t}\succeq\mathbf{R}_{P}$ establishes
$s_{P}$. Then
\begin{equation}
(\mathbf{S}^{P}_{t})^{-1}-(\widehat{\mathbf{S}}^{P}_{t})^{-1}=(\mathbf{S}^{P}_{t})^{-1}(\widehat{\mathbf{S}}^{P}_{t}-\mathbf{S}^{P}_{t})(\widehat{\mathbf{S}}^{P}_{t})^{-1}.
\end{equation}
Taking spectral norms proves the first bound in the proposition.

For the Gram blocks,
\begin{equation}
\adjustbox{max width=.88\linewidth}{\ensuremath{\mathbf{G}_{t,ij}-\widehat{\mathbf{G}}_{t,ij}=\mathbf{D}^{\mathsf{T}}_{t,i}\mathbf{W}\widehat{\mathbf{C}}_{t,j}+\widehat{\mathbf{C}}^{\mathsf{T}}_{t,i}\mathbf{W}\mathbf{D}_{t,j}+\mathbf{D}^{\mathsf{T}}_{t,i}\mathbf{W}\mathbf{D}_{t,j}.}}\label{eq:graph_edge-1-2}
\end{equation}
Using $\|\mathbf{A}^{\mathsf{T}}\mathbf{B}\|_{F}\leq\|\mathbf{A}\|_{F}\|\mathbf{B}\|_{F}$
and $\|\mathbf{W}^{1/2}\mathbf{D}_{t,j}\|_{F}\leq\omega\epsilon_{t,j}$
yields $\beta_{t,ij}$ in (\ref{eq:Two_bounded-1}). In particular,
$\beta_{t,ij}=\beta_{t,ji}\geq0$.

For the inverse contribution, $\mathbf{G}^{P}_{t}\succeq0$ implies
\begin{equation}
\adjustbox{max width=.89\linewidth}{\ensuremath{|\mathrm{tr}[((\mathbf{S}^{P}_{t})^{-1}-(\widehat{\mathbf{S}}^{P}_{t})^{-1})\mathbf{G}^{P}_{t}]|\leq\|(\mathbf{S}^{P}_{t})^{-1}-(\widehat{\mathbf{S}}^{P}_{t})^{-1}\|_{2}\mathrm{tr}(\mathbf{G}^{P}_{t}).}}\label{eq:graph_edge-1-2-1}
\end{equation}
Moreover,
\begin{equation}
\mathrm{tr}(\mathbf{G}^{P}_{t})=\|\mathbf{W}^{1/2}\mathbf{C}^{P}_{t}\|^{2}_{F}\leq(\|\mathbf{W}^{1/2}\widehat{\mathbf{C}}^{P}_{t}\|_{F}+\omega\epsilon_{t}(P))^{2}.
\end{equation}
Combining the last two inequalities with (\ref{eq:conditioning stability})
establishes the bound $\kappa_{t}(P)$ in (\ref{eq:Two_bounded-1}),
thereby completing the proof.

\subsection{Proof of Corollary 1}\label{Proof of Corollary 1}

Let $\mathbf{F}_{t,ij}=\mathbf{G}_{t,ij}-\widehat{\mathbf{G}}_{t,ij}$
be a shared Gram-block error, and define the innovation-inverse contribution
\begin{equation}
\zeta_{t}(P)=\mathrm{tr}(((\mathbf{S}^{P}_{t})^{-1}-(\widehat{\mathbf{S}}^{P}_{t})^{-1})\mathbf{G}^{P}_{t}).
\end{equation}
The shared Gram-error contribution satisfies
\begin{equation}
\adjustbox{max width=.89\linewidth}{\ensuremath{\begin{aligned} & J_{t}(P)-J_{t}(Q)-(\widehat{J}_{t}(P)-\widehat{J}_{t}(Q))\\
 & =\sum_{i,j\in\mathcal{P\cup Q}}\left\langle [(\widehat{\mathbf{S}}^{P}_{t})^{-1}]_{ij}-[(\widehat{\mathbf{S}}^{Q}_{t})^{-1}]_{ij},\mathbf{F}_{t,ij}\right\rangle +\zeta_{t}(P)-\zeta_{t}(Q),
\end{aligned}
}}\label{eq:graph_edge-1-2-3}
\end{equation}
where $\langle\mathbf{A},\mathbf{B}\rangle=\mathrm{tr}(\mathbf{A}^{\mathsf{T}}\mathbf{B})$
is the Frobenius inner product. Applying the triangle and Cauchy-Schwarz
inequalities gives
\begin{equation}
\ensuremath{\begin{aligned} & \left|J_{t}(P)-J_{t}(Q)-(\widehat{J}_{t}(P)-\widehat{J}_{t}(Q))\right|\\
 & \le\sum_{i,j\in\mathcal{P\cup Q}}\left\Vert \mathbf{V}_{t,ij}(P)-\mathbf{V}_{t,ij}(Q)\right\Vert _{\mathrm{F}}\beta_{t,ij}\\
 & \hspace*{0.5cm}+\kappa_{t}(P)+\kappa_{t}(Q)=e_{t}(P,Q).
\end{aligned}
}\label{eq:graph_edge-1-2-4}
\end{equation}
For the separate route-value allowance
\begin{equation}
a_{t}(P)=\sum_{i,j\in P}\|[(\widehat{\mathbf{S}}^{P}_{t})^{-1}]_{ij}\|_{\mathrm{F}}\beta_{t,ij}+\kappa_{t}(P),
\end{equation}
the triangle inequality and $\beta_{t,ij}\ge0$ yield
\begin{equation}
e_{t}(P,Q)\le a_{t}(P)+a_{t}(Q).
\end{equation}
This completes the proof.

\subsection{Conditional-Response Error and Gram Evaluation}\label{Conditional-Response Error}

For a nonempty prefix $P$ and $j\notin P$, define
\begin{equation}
\begin{aligned}\mathbf{K}^{P}_{t,j} & =\mathbf{E}_{P}\mathbf{C}_{t,j}, & \mathbf{B}^{P}_{t,j} & =(\mathbf{S}^{P}_{t})^{-1}\mathbf{K}^{P}_{t,j},\\
\widehat{\mathbf{K}}^{P}_{t,j} & =\mathrm{sym}(\mathbf{E}_{P}\widehat{\mathbf{C}}_{t,j}), & \widehat{\mathbf{B}}^{P}_{t,j} & =(\widehat{\mathbf{S}}^{P}_{t})^{-1}\widehat{\mathbf{K}}^{P}_{t,j}.
\end{aligned}
\end{equation}
Then $\widehat{\mathbf{C}}^{P}_{t,j}=\widehat{\mathbf{C}}_{t,j}-\widehat{\mathbf{C}}^{P}_{t}\widehat{\mathbf{B}}^{P}_{t,j}$.
A computable bound is given by
\begin{equation}
\begin{aligned}\delta_{t,j}(P) & =\omega(\epsilon_{t,j}+\epsilon_{t}(P)\|\widehat{\mathbf{B}}^{P}_{t,j}\|_{F})\\
 & \hspace*{0.4cm}+\frac{\|\mathbf{W}^{1/2}\widehat{\mathbf{C}}^{P}_{t}\|_{F}+\omega\epsilon_{t}(P)}{s_{P}}\\
 & \hspace*{0.4cm}\times(\frac{\epsilon_{t}(P)+\epsilon_{t,j}}{2}+(\epsilon_{t}(P)+\gamma_{t})\|\widehat{\mathbf{B}}^{P}_{t,j}\|_{F}).
\end{aligned}
\end{equation}
To prove it, covariance symmetry gives
\begin{equation}
\|\mathbf{K}^{P}_{t,j}-\widehat{\mathbf{K}}^{P}_{t,j}\|_{F}\leq\frac{\epsilon_{t}(P)+\epsilon_{t,j}}{2}.
\end{equation}
The conditioning-coefficient difference satisfies
\begin{equation}
\adjustbox{max width=.88\linewidth}{\ensuremath{\mathbf{B}^{P}_{t,j}-\widehat{\mathbf{B}}^{P}_{t,j}=(\mathbf{S}^{P}_{t})^{-1}(\mathbf{K}^{P}_{t,j}-\widehat{\mathbf{K}}^{P}_{t,j}+(\widehat{\mathbf{S}}^{P}_{t}-\mathbf{S}^{P}_{t})\widehat{\mathbf{B}}^{P}_{t,j}).}}\label{eq:graph_edge-1-2-2}
\end{equation}
Using Proposition 2 yields
\begin{equation}
\adjustbox{max width=.88\linewidth}{\ensuremath{||\mathbf{B}^{P}_{t,j}-\widehat{\mathbf{B}}^{P}_{t,j}||_{F}\leq\frac{\frac{1}{2}(\epsilon_{t}(P)+\epsilon_{t,j})+(\epsilon_{t}(P)+\gamma_{t})\|\widehat{\mathbf{B}}^{P}_{t,j}\|_{F}}{s_{P}}.}}\label{eq:graph_edge-1-2-2-1}
\end{equation}
Finally, the conditional-response error decomposes as
\begin{equation}
\begin{aligned}\mathbf{C}^{P}_{t,j}-\widehat{\mathbf{C}}^{P}_{t,j}={} & \mathbf{D}_{t,j}-\mathbf{D}^{P}_{t}\widehat{\mathbf{B}}^{P}_{t,j}-\mathbf{C}^{P}_{t}(\mathbf{B}^{P}_{t,j}-\widehat{\mathbf{B}}^{P}_{t,j}).\end{aligned}
\end{equation}
Applying the weighted Frobenius norm and the preceding bounds establishes
\begin{equation}
\left\Vert \mathbf{W}^{1/2}(\mathbf{C}^{P}_{t,j}-\widehat{\mathbf{C}}^{P}_{t,j})\right\Vert _{F}\leq\delta_{t,j}(P).
\end{equation}
For the empty prefix, only the root-response error remains, giving
$\delta_{t,j}(\varnothing)=\omega\epsilon_{t,j}$.

\subsection{Proof of Lemma 1}\label{Proof of Lemma 1}

Consider a feasible continuation $T$ of prefix $P$, define $\mathbf{C}^{P}_{t,T}=\boldsymbol{\Sigma}^{P}_{t}\mathbf{E}^{\mathsf{T}}_{T}$
and $\mathbf{S}^{P}_{t}(T)=\mathbf{E}_{T}\boldsymbol{\Sigma}^{P}_{t}\mathbf{E}^{\mathsf{T}}_{T}+\mathbf{R}_{T}.$
Since $\mathbf{S}^{P}_{t}(T)\succeq\mathbf{R}_{T}\succ0$,
\begin{equation}
\begin{aligned}J_{t}(P\oplus T)-J_{t}(P) & =\operatorname{tr}\left[(\mathbf{S}^{P}_{t}(T))^{-1}(\mathbf{C}^{P}_{t,T})^{\mathsf{T}}\mathbf{W}\mathbf{C}^{P}_{t,T}\right]\\
 & \leq\sum_{j\in T}\|\mathbf{W}^{1/2}\mathbf{C}^{P}_{t,j}\mathbf{R}^{-1/2}_{j}\|^{2}_{F}\\
 & \leq\sum_{j\in T}\overline{b}_{t,j}(P).
\end{aligned}
\end{equation}
The last inequality uses the conditional-response error and $\ensuremath{\|\mathbf{R}^{-1/2}_{j}\|^{2}_{2}=1/\lambda_{\min}(\mathbf{R}_{j})}.$
Every feasible continuation contains at most $r_{P}$ locations drawn
from $\mathcal{A}_{t}(P)$. Replacing its sum by the largest at most
$r_{P}$ complete the proof.

\subsection{Proof of Proposition 3}\label{Proof of Proposition 3}

For any prefix $P$ excluded by the certified continuation bound,
every feasible completion $P'$ satisfies $J_{t}(P')\leq J_{t}(Q^{0}_{t}).$
Since each incumbent replacement is certified by a positive lower
comparison bound, the exact incumbent value is nondecreasing. Therefore,
\begin{equation}
J_{t}(P')\leq J_{t}(Q^{0}_{t})\leq J_{t}(Q_{t}).
\end{equation}
Likewise, any terminal route previously excluded by a certified comparison
against an earlier incumbent remains dominated by $Q_{t}$. By complete
route coverage, every remaining terminal route belongs to $\mathcal{T}_{t}$
or has been certified as dominated. For each $P\in\mathcal{T}_{t}$,
the valid comparison bound gives
\begin{equation}
J_{t}(P)-J_{t}(Q_{t})\leq\overline{\Delta}_{t}(P,Q_{t})\leq g_{t}.
\end{equation}
Taking the maximum over the declared route family establishes (\ref{eq:Refinement-compatible}).
Finally, numerical refinement changes only the approximate responses
and their certificates. Previously certified dominance relations therefore
remain valid, and excluded branches need not be revisited. This completes
the proof.

\bibliographystyle{plainnat}
\bibliography{Bib/IEEEabrv,Bib/RadioMap}

\end{document}

%% file: manuscript-commands.tex
\usepackage[acronym]{glossaries}
\glsdisablehyper
\newcommand{\newac}{\newacronym}
\newcommand{\ac}{\gls}

\newcommand{\acpl}{\glspl}

\input{JCgroup_acronym.tex}
\newac{dl}{DL}{deep learning}
\newac{rt}{RT}{ray tracing}
\newac{nn}{NN}{neural network}
\newac{tx}{TX}{transmitter}
\newac{rx}{RX}{receiver}
\newtheorem{prop}{Proposition}
\newtheorem{cor}{Corollary}
\newtheorem{lem}{Lemma}

%% file: JCgroup_acronym.tex
\newac{speb}{SPEB}{square position error bound}
\newac[plural=EFIMs,firstplural=Fisher information matrices (EFIMs)]{efim}{EFIM}{Fisher information matrix}
\newac{ne}{NE}{Nash equilibrium}
\newac{mse}{MSE}{mean squared error}
\newac{toa}{TOA}{time-of-arrival}
\newac{snr}{SNR}{signal-to-noise ratio}
\newac{lan}{LAN}{local area network}
\newac{psd}{PSD}{positive semidefinite}
\newac{pd}{PD}{positive definite}
\newac{wrt}{w.r.t.}{with respect to}
\newac{lhs}{L.H.S.}{left hand side}
\newac{wp1}{w.p.1}{with probability 1}
\newac{kkt}{KKT}{Karush-Kuhn-Tucker}
\newac{wlog}{w.l.o.g.}{without loss of generality}
\newac{mle}{MLE}{maximum likelihood estimation}
\newac{gps}{GPS}{global positioning system}
\newac{rssi}{RSSI}{received signal strength indicator}
\newac{mimo}{MIMO}{multiple-input multiple-output}
\newac{csi}{CSI}{channel state information}
\newac{fdd}{FDD}{frequency division duplexing}
\newac{ms}{MS}{mobile station}
\newac{bs}{BS}{base station}
\newac{d2d}{D2D}{device-to-device}
\newac{slnr}{SLNR}{signal-to-interference-leakage-and-noise-ratio}
\newac{ula}{ULA}{uniform linear antenna array}
\newac{pas}{PAS}{power angular spectrum}
\newac{mmse}{MMSE}{minimum mean square error}
\newac{zf}{ZF}{zero-forcing}
\newac{rzf}{RZF}{regularized zero-forcing}
\newac{as}{AS}{angular spread}
\newac{aod}{AOD}{angle of departure}
\newac{iid}{i.i.d.}{independent and identically distributed} 
\newac{sinr}{SINR}{signal-to-interference-and-noise ratio}
\newac{tdd}{TDD}{time-division duplex}
\newac{rvq}{RVQ}{random vector quantization}
\newac{rhs}{R.H.S.}{right hand side}
\newac{mrc}{MRC}{maximum ratio combining}
\newac{cdf}{CDF}{cumulative distribution function}
\newac{a.s.}{a.s.}{almost surely}
\newac{los}{LOS}{line-of-sight}
\newac{jsdm}{JSDM}{joint spatial division and multiplexing}
\newac{map}{MAP}{maximum a posteriori}
\newac{klt}{KLT}{Karhunen-Lo\`eve Transform}
\newac{lbe}{LBE}{link bargaining equilibrium}
\newac{se}{SE}{Stackelberg equilibrium}
\newac{uav}{UAV}{unmanned aerial vehicle}
\newac{nlos}{NLOS}{non-line-of-sight}
\newac{pdf}{PDF}{probability density function}
\newac{em}{EM}{expectation-maximization}
\newac{knn}{KNN}{$k$-nearest neighbor}
\newac{svd}{SVD}{singular value decomposition}
\newac{nmf}{NMF}{non-negative matrix factorization}
\newac{umf}{UMF}{unimodality-constrained matrix factorization}
\newac{rmse}{RMSE}{rooted mean squared error}
\newac{olos}{OLOS}{obstructed line-of-sight}
\newac{mmw}{mmW}{millimeter wave}
\newac{ber}{BER}{bit error rate}
\newac{rss}{RSS}{received signal strength}
\newac{lp}{LP}{linear program}
\newac{ufw}{U-FW}{unimodal Frank-Wolfe}
\newac{utf}{UTF}{unimodality-constrained tensor factorization}
\newac{fw}{FW}{Frank-Wolfe}
\newac{iot}{IoT}{Internet-of-Things}
\newac{mae}{MAE}{mean absolute error}
\newac{crb}{CRB}{Cram\'er-Rao bound}
\newac{aoa}{AoA}{angle of arrival}
\newac{wcl}{WCL}{weighted centroid localization}

%% file: Bib/RadioMap.bib
@ARTICLE{WangSunZem:J25,
  author={Wang, Yixian and Sun, Geng and Sun, Zemin and Wang, Jiacheng and Li, Jiahui and Zhao, Changyuan and Wu, Jing and Liang, Shuang and Yin, Minghao and Wang, Pengfei and Niyato, Dusit and Sun, Sumei and In Kim, Dong},
  journal={IEEE Trans. Cognit. Commun. Netw.},
  title={Toward Realization of Low-Altitude Economy Networks: Core Architecture, Integrated Technologies, and Future Directions},
  year={2025},
  volume={11},
  number={5},
  pages={2788-2820},
  doi={10.1109/TCCN.2025.3601015}}

@ARTICLE{ZenChenJ:24,
  author={Zeng, Yong and Chen, Junting and Xu, Jie and Wu, Di and Xu, Xiaoli and Jin, Shi and Gao, Xiqi and Gesbert, David and Cui, Shuguang and Zhang, Rui},
  journal={IEEE Commun. Surveys Tuts.},
  title={A Tutorial On Environment-Aware Communications via Channel Knowledge Map for {6G}},
  year={2024},
  volume={},
  number={},
  pages={1-1},
  doi={10.1109/COMST.2024.3364508}}

@ARTICLE{WuZengJin:J24,
  author={Wu, Di and Zeng, Yong and Jin, Shi and Zhang, Rui},
  journal={IEEE Trans. Wireless Commun.},
  title={Environment-Aware Hybrid Beamforming by Leveraging Channel Knowledge Map},
  year={2024},
  volume={23},
  number={5},
  pages={4990-5005},
  doi={10.1109/TWC.2023.3323941}}

@ARTICLE{LiBoChen:J24,
  author={Li, Bowen and Chen, Junting},
  journal={IEEE Trans. Wireless Commun.},
  title={Radio Map-Assisted Approach for Interference-Aware Predictive {UAV} Communications},
  year={2024},
  volume={23},
  number={11},
  pages={16725-16741},
  doi={10.1109/TWC.2024.3446240}}

@ARTICLE{SunChe:J24,
  author={Sun, Hao and Chen, Junting},
  journal={IEEE Trans. Signal Process.},
  title={Integrated Interpolation and Block-Term Tensor Decomposition for Spectrum Map Construction},
  year={2024},
  volume={72},
  number={},
  pages={3896-3911},
  doi={10.1109/TSP.2024.3439513}}

@ARTICLE{ShaoZhaLi:J26,
  author={Shao, Xiaoyan and Zhang, Wence and Li, Haohan and Shao, Zhichao and Zhang, Zhiguang and Bao, Xu},
  journal={IEEE Internet Things J.},
  title={A Low-Complexity Channel Knowledge Map Construction Based on Environmental Partitioning and Interpolation Weight Learning},
  year={2026},
  volume={13},
  number={11},
  pages={24044-24056},
  doi={10.1109/JIOT.2026.3672412}}

@ARTICLE{TimiShre:J24,
  author={Timilsina, Subash and Shrestha, Sagar and Fu, Xiao},
  journal={IEEE Trans. Signal Process.},
  title={Quantized Radio Map Estimation Using Tensor and Deep Generative Models},
  year={2024},
  volume={72},
  number={},
  pages={173-189},
  doi={10.1109/TSP.2023.3336179}}

@ARTICLE{ZhaFu:J20,
  author={Zhang, Guoyong and Fu, Xiao and Wang, Jun and Zhao, Xi-Le and Hong, Mingyi},
  journal={IEEE Trans. Signal Process.},
  title={Spectrum Cartography via Coupled Block-Term Tensor Decomposition},
  year={2020},
  volume={68},
  number={},
  pages={3660-3675},
  doi={10.1109/TSP.2020.2993530}}

@ARTICLE{WangTaoNan:J25,
  author={Wang, Xiucheng and Tao, Keda and Cheng, Nan and Yin, Zhisheng and Li, Zan and Zhang, Yuan and Shen, Xuemin},
  journal={IEEE Trans. Cognit. Commun. Netw.},
  title={RadioDiff: An Effective Generative Diffusion Model for Sampling-Free Dynamic Radio Map Construction},
  year={2025},
  volume={11},
  number={2},
  pages={738-750},
  doi={10.1109/TCCN.2024.3504489}}

@ARTICLE{WangqianChen:J24,
  author={Chen, Wangqian and Chen, Junting},
  journal={IEEE Trans. Wireless Commun.},
  title={Diffraction and Scattering Aware Radio Map and Environment Reconstruction Using Geometry Model-Assisted Deep Learning},
  year={2024},
  volume={23},
  number={12},
  pages={19804-19819},
  doi={10.1109/TWC.2024.3487225}}

@ARTICLE{JaenCair:J26,
  author={Jaensch, Fabian and Caire, Giuseppe and Demir, Begüm},
  journal={IEEE Trans. Wireless Commun.},
  title={Radio Map Prediction From Aerial Images and Application to Coverage Optimization},
  year={2026},
  volume={25},
  number={},
  pages={308-320},
  doi={10.1109/TWC.2025.3583171}}

@ARTICLE{TegaRomer:J22,
  author={Teganya, Yves and Romero, Daniel},
  journal={IEEE Trans. Wireless Commun.},
  title={Deep Completion Autoencoders for Radio Map Estimation},
  year={2022},
  volume={21},
  number={3},
  pages={1710-1724},
  doi={10.1109/TWC.2021.3106154}}

@INPROCEEDINGS{ZhangZhouYou:C25,
  author={Zhang, Na and Zhou, Rui and You, Ming-Yi and Pu, Wenqiang},
  booktitle={Proc. European Signal Process. Conf. (EUSIPCO)},
  title={Adaptive Sampling for Online UAV-Aided Radio Map Sensing},
  year={2025},
  volume={},
  number={},
  pages={985-989},
  doi={10.23919/EUSIPCO63237.2025.11226752}}

@ARTICLE{PolySadYe:J24,
  author={Polyzos, Konstantinos D. and Sadeghi, Alireza and Ye, Wei and Sleder, Steven and Houssou, Kodjo and Calder, Jeff and Zhang, Zhi-Li and Giannakis, Georgios B.},
  journal={IEEE Trans. Wireless Commun.},
  title={Bayesian Active Learning for Sample Efficient {5G} Radio Map Reconstruction},
  year={2024},
  volume={23},
  number={12},
  pages={19382-19396},
  doi={10.1109/TWC.2024.3483112}}

@ARTICLE{ShreRomeChe:J23,
  author={Shrestha, Raju and Romero, Daniel and Chepuri, Sundeep Prabhakar},
  journal={IEEE Trans. Wireless Commun.},
  title={Spectrum Surveying: Active Radio Map Estimation With Autonomous UAVs},
  year={2023},
  volume={22},
  number={1},
  pages={627-641},
  doi={10.1109/TWC.2022.3197087}}

@INPROCEEDINGS{SunShicXiang:C26,
  author={Sun, Hao and Liu, Shicong and Yu, Xianghao and Sun, Ying},
  booktitle={Proc. Int. Conf. Wireless Commun. Signal Process. (ICASSP)},
  title={Flow Matching-Based Active Learning for Radio Map Construction with Low-Altitude UAVs},
  year={2026},
  volume={},
  number={},
  pages={21406-21410},
  doi={10.1109/ICASSP55912.2026.11463842}}

@ARTICLE{gu2025generat,
  author={Yuntong Gu and Xiangming meng and Zhiyuan Lin and Sheng Wu and Linling Kuang},
  journal={arXiv:2512.20108},
  title={Generative Bayesian Spectrum Cartography: Unified Reconstruction and Active Sensing via Diffusion Models},
  year={2025},
  volume={},
  number={},
  pages={},
  doi={https://arxiv.org/abs/2512.20108}}

@ARTICLE{ChenWangGuo:J25,
  author={Chen, Zhibo and Wang, Heng and Guo, Daoxing},
  journal={IEEE Wireless Commun. Lett.},
  title={{3D} Radio Map Estimation Based on Active Measurement Trajectory Selection},
  year={2025},
  volume={14},
  number={7},
  pages={1884-1888},
  doi={10.1109/LWC.2025.3557556}}

@INPROCEEDINGS{ZheXiao:C25,
  author={Zheng, Zhihao and Xiao, Limin and Zhao, Ming and Li, Yunzhou},
  booktitle={Proc. IEEE Int. Workshop Mach. Learn. Signal Process. (MLSP)},
  title={Efficient Indoor Radio Map Prediction with Improved Transformers and Active Sampling Strategies},
  year={2025},
  volume={},
  number={},
  pages={1-6},
  doi={10.1109/MLSP62443.2025.11204304}}

@ARTICLE{LiYaoAls:J26,
  author={Li, Mingxu and Shi, Yao and Alsusa, Emad and Zhang, Deyou and Deng, Jie and Su, Nanchi and You, Xiaohu},
  journal={IEEE Internet Things J.},
  title={End-to-End UAV-Enabled Adaptive {3D} Radio Mapping via Joint Optimization of Sparse Sampling and Reconstruction},
  year={2026},
  volume={13},
  number={11},
  pages={25100-25113},
  doi={10.1109/JIOT.2026.3676550}}

@ARTICLE{ZhaoHaoHua:J24,
  author={Zhao, Haitao and Hao, Qing and Huang, Hao and Gui, Guan and Ohtsuki, Tomoaki and Sari, Hikmet and Adachi, Fumiyuki},
  journal={IEEE Trans. Veh. Technol.},
  title={Online Trajectory Optimization for Energy-Efficient Cellular-Connected {UAV}s With Map Reconstruction},
  year={2024},
  volume={73},
  number={3},
  pages={3445-3456},
  doi={10.1109/TVT.2023.3323349}}

@misc{sionnarttechnicalreport,
      title={Sionna {RT}: Technical Report},
      author={Fayçal Aït Aoudia and Jakob Hoydis and Merlin Nimier-David and Baptiste Nicolet and Sebastian Cammerer and Alexander Keller},
      year={2025},
      eprint={2504.21719},
      archivePrefix={arXiv},
      primaryClass={cs.IT},
      url={https://arxiv.org/abs/2504.21719.},
}

@INPROCEEDINGS{ChenWChen:C25,
  author={Chen, Wangqian and Chen, Junting},
  booktitle={Proc. IEEE Int. Conf. Commun. (ICC)},
  title={Learning Blockage and Reflection Geometry for {MIMO} Beam Map Construction},
  year={2025},
  volume={},
  number={},
  pages={2321-2326},
  doi={10.1109/ICC52391.2025.11161713}}

@ARTICLE{ChenZhuWang:J26,
  author={Chen, Yiran and Zhu, Qiuming and Wang, Jie and Jia, Ziye and Wang, Xuan and Lin, Zhipeng and Huang, Yang and Wu, Qihui and Briso-Rodríguez, César},
  journal={IEEE Trans. Cogn. Commun. Netw.},
  title={{UAV}-Aided Efficient Informative Path Planning for Autonomous {3D} Spectrum Mapping},
  year={2026},
  volume={12},
  number={},
  pages={1664-1677},
  doi={10.1109/TCCN.2025.3603736}}

@INPROCEEDINGS{LuGao:C25,
  author={Lu, Wenlihan and Gao, Shijian and Wen, Miaowen and Liang, Yuxuan and Yang, Liuqing and Chae, Chan-Byoung and Poor, H. Vincent},
  booktitle={Int. Conf. Wireless Commun. and Signal Process. (WCSP)},
  title={Bayesian-Driven Graph Reasoning for Active Radio Map Construction},
  year={2025},
  volume={},
  number={},
  pages={1-6},
  doi={10.1109/WCSP68525.2025.1010321}}

@book{rasmussen2006gaussian,
  author    = {Carl Edward Rasmussen and Christopher K. I. Williams},
  title     = {Gaussian Processes for Machine Learning},
  publisher = {The MIT Press},
  address   = {Cambridge, MA, USA},
  year      = {2006},
  isbn      = {0-262-18253-X},
  url       = {https://gaussianprocess.org/gpml/}
}

@article{LinTao:J2026,
  author  = {Junyi Lin and Tao Wu and Hongjun Wang and
             Chao Dong and Zhexian Shen and Mingjin Zhang and
             Jiannong Cao},
  title   = {{UAV}-Enabled Adaptive Radio Sampling for High-Quality Radio Environment Map Construction in Unknown Environments},
  journal = {IEEE Trans. Mobile Comput.},
  year    = {2026},
  doi     = {10.1109/TMC.2026.3737955},
  note    = {Early Access}
}

@article{ChenXiao:J2025,
  author  = {Xinwei Chen and Xiaofeng Zhong and Zijian Zhang and Linglong Dai and Shidong Zhou},
  title   = {High-Efficiency Urban 3-D Radio Map Estimation Based on Sparse Measurements},
  journal = {IEEE Trans. Veh. Technol.},
  volume  = {74},
  number  = {10},
  pages   = {16488--16493},
  year    = {2025},
  doi     = {10.1109/TVT.2025.3573595}
}
